\documentclass[11pt,letterpaper]{article}
\usepackage[lmargin=1.0in,rmargin=1.0in,bottom=1.0in,top=1.0in,
twoside=False]{geometry}

\usepackage{fullpage,amssymb,amsmath}
\usepackage{graphicx,algorithmic,algorithm}%, verbatim}
\usepackage{enumerate}
\usepackage{tikz}
\usetikzlibrary{shapes}
\usepackage[T1]{fontenc}

\usepackage{xspace}
\usepackage{xcolor}
\usepackage{mathtools}
\usepackage{microtype}
\usepackage{amsfonts}
\usepackage{comment}
\usepackage[english]{babel}
\usepackage{mathrsfs}

\usepackage{fontaxes}
\usepackage{trimspaces}
\usepackage{nccfoots}
\usepackage{setspace}
\usepackage{inconsolata}
\usepackage{libertine}

\usepackage{enumitem}
\usepackage{todonotes}

\definecolor{blue}{rgb}{0.1,0.2,0.5}
\definecolor{brown}{rgb}{0.6,0.6,0.2}

\usepackage[amsmath,thmmarks,hyperref]{ntheorem}
\usepackage{cleveref}

\crefformat{page}{#2page~#1#3}%
\Crefformat{page}{#2Page~#1#3}%
\crefformat{equation}{#2(#1)#3}%
\Crefformat{equation}{#2(#1)#3}%
\crefformat{figure}{#2Figure~#1#3}%
\Crefformat{figure}{#2Figure~#1#3}%
\crefformat{section}{#2Section~#1#3}
\Crefformat{section}{#2Section~#1#3}
\crefformat{chapter}{#2Chapter~#1#3}
\Crefformat{chapter}{#2Chapter~#1#3}
\crefformat{chapter*}{#2Chapter~#1#3}
\Crefformat{chapter*}{#2Chapter~#1#3}
\crefformat{part}{#2Part~#1#3}
\Crefformat{part}{#2Part~#1#3}
\crefformat{enumi}{#2(#1)#3}
\Crefformat{enumi}{#2(#1)#3}

\usepackage{enumerate}

\usepackage{latexsym}

\theoremnumbering{arabic}
\theoremstyle{plain}
\theoremsymbol{}
\theorembodyfont{\itshape}
\theoremheaderfont{\normalfont\bfseries}
\theoremseparator{.}
\theorempreskip{4pt}
\theorempostskip{3pt}

\newtheorem{theorem}{Theorem}
\crefformat{theorem}{#2Theorem~#1#3}
\Crefformat{theorem}{#2Theorem~#1#3}

\newcommand{\newtheoremwithcrefformat}[2]{%
  \newtheorem{#1}[theorem]{#2}%
  \crefformat{#1}{##2\MakeUppercase#1~##1##3}%
  \Crefformat{#1}{##2\MakeUppercase#1~##1##3}%
}
\newcommand{\newseptheoremwithcrefformat}[2]{%
  \newtheorem{#1}{#2}%
  \crefformat{#1}{##2\MakeUppercase#1~##1##3}%
  \Crefformat{#1}{##2\MakeUppercase#1~##1##3}%
}

\newtheoremwithcrefformat{lemma}{Lemma}
\newtheoremwithcrefformat{proposition}{Proposition}
\newtheoremwithcrefformat{observation}{Observation}
\newtheoremwithcrefformat{conjecture}{Conjecture}
\newtheoremwithcrefformat{corollary}{Corollary}
\newseptheoremwithcrefformat{claim}{Claim}
\theorembodyfont{\upshape}
\newtheoremwithcrefformat{example}{Example}
\newtheoremwithcrefformat{remark}{Remark}
\newseptheoremwithcrefformat{definition}{Definition}

\theoremstyle{nonumberplain}
\theoremheaderfont{\scshape}
\theorembodyfont{\normalfont}
\theoremsymbol{\ensuremath{\square}}
\newtheorem{proof}{Proof}

\theoremsymbol{\ensuremath{\lrcorner}}

\def\cqedsymbol{\ifmmode$\lrcorner$\else{\unskip\nobreak\hfil
\penalty50\hskip1em\null\nobreak\hfil$\lrcorner$
\parfillskip=0pt\finalhyphendemerits=0\endgraf}\fi} 

\newcommand{\cqed}{\renewcommand{\qed}{\cqedsymbol}}

\newcommand{\Oh}{\mathcal{O}}
\newcommand{\Pp}{\mathcal{P}}

\newcommand{\Cc}{\mathscr{C}}
\newcommand{\Dd}{\mathscr{D}}

\newcommand{\rDomSet}{{\sc{\mbox{Distance-$r$} Dominating Set}}\xspace}
\newcommand{\domSet}{{\sc{Dominating Set}}\xspace}
\newcommand{\rIndSet}{{\sc{\mbox{Distance-$r$} Independent Set}}\xspace}

\newcommand{\profile}{\mathrm{profile}}
\newcommand{\tn}[1]{{\tiny{#1}}}

\newcommand{\set}[1]{\{#1\}}
\newcommand{\setof}[2]{\set{#1\colon #2}}

\newcommand{\numprofiles}[3]{\nu^{#2}_{#1}(#3)}
\newcommand{\eps}{\varepsilon}
\newcommand{\poly}{\mathrm{poly}}
\newcommand{\nats}{\mathbb{N}}

\newcommand{\Pc}{{\mathscr P}}
\newcommand{\Gg}{\mathcal{G}}

\renewcommand{\subset}{\subseteq}

\newcommand{\N}{\mathbb{N}}

\newcommand{\tup}[1]{\overline{#1}}
\renewcommand{\phi}{\varphi}
\renewcommand{\epsilon}{\varepsilon}

\newcommand{\dist}{\mathrm{dist}}

\newcommand{\Ff}{\mathcal{F}}

\renewcommand{\geq}{\geqslant}
\renewcommand{\leq}{\leqslant}
\renewcommand{\ge}{\geqslant}
\renewcommand{\le}{\leqslant}
\renewcommand{\setminus}{-}

\hypersetup{
    colorlinks = true,
	linkcolor={blue}, 
	citecolor={brown}
}

\begin{document}
\title{Progressive Algorithms for Domination and Independence\thanks{The work of M.\ Pilipczuk and S.\ Siebertz is supported by the National Science Centre of 
Poland via POLONEZ grant agreement UMO-2015/19/P/ST6/03998, 
which has received funding from the European Union's Horizon 2020 research and 
innovation programme (Marie Sk\l odowska-Curie grant agreement No.\ 665778).}}

\author{
Grzegorz Fabia\'nski\thanks{
Faculty of Mathematics, Informatics, and Mechanics, University of Warsaw, Poland,\newline
\texttt{grzegorz.fabianski@students.mimuw.edu.pl}}
\and
Micha\l~Pilipczuk\thanks{
  Institute of Informatics, University of Warsaw, Poland, \texttt{michal.pilipczuk@mimuw.edu.pl}.
}
\and
Sebastian Siebertz\thanks{
  Institut f\"ur Informatik
Humboldt-Universit\"at zu Berlin \texttt{siebertz@informatik.hu-berlin.de}.
}
\and
Szymon Toru\'nczyk\thanks{
Institute of Informatics, University of Warsaw, Poland
\texttt{szymtor@mimuw.edu.pl}
}}

\maketitle

\begin{abstract}
  We consider a generic algorithmic paradigm that we call
  \emph{progressive exploration}, which can be used to develop simple and efficient parameterized
  graph algorithms.  We identify two model-theoretic properties that
  lead to efficient progressive algorithms, namely variants of the
  \emph{Helly property} and \emph{stability}.  We demonstrate our
  approach by giving linear-time fixed-parameter algorithms for the \rDomSet
  problem (parameterized by the solution
  size) in a wide variety of restricted graph classes, 
  such as powers of nowhere dense
  classes, map graphs, and (for $r=1$) biclique-free graphs.
  Similarly, for the \rIndSet problem the technique can be used to give a linear-time fixed-parameter algorithm on any nowhere dense class.
  Despite the simplicity of the method, in several cases our results extend known boundaries of tractability for the considered problems and improve the best known running times.
\end{abstract}

\begin{picture}(0,0) \put(410,-280)
{\hbox{\includegraphics[scale=0.25]{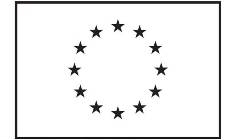}}} \end{picture} 
\vspace{-0.8cm}

\section{Introduction}\label{sec:intro}

It is widely believed that many
important algorithmic graph problems cannot be solved efficiently on
general graphs. Consequently, a natural question is to identify the
most general classes of graphs on which a given problem can be solved
efficiently. Structural graph theory offers a wealth of concepts that
can be used to design efficient algorithms for generally intractable
problems on restricted graph classes. An important result in this area
states that every property of graphs expressible in monadic
second-order logic can be tested in linear time on every class
of bounded treewidth~\cite{courcelle1990monadic}. 
Similarly, every property expressible in first-order logic can be tested in almost 
linear time on every nowhere dense graph class~\cite{GroheKS17}.

%\vspace{-0.3pt}
Nowhere denseness is an abstract notion of uniform sparseness in graphs, 
which is the foundational definition of the theory of sparse graphs; see
the monograph of Ne\v{s}et\v{r}il and Ossona de Mendez~\cite{sparsity} for an overview.
Formally, a graph class $\Cc$ is {\em{nowhere dense}} if for every $r\in \N$, one
cannot obtain arbitrary large cliques by contracting disjoint connected subgraphs of radius at most $r$ in graphs from $\Cc$.
Many well-studied classes of sparse graphs are nowhere dense, for instance the class of planar
graphs, any class of graphs with a fixed bound on the maximum degree, or any class of graphs excluding a fixed (topological) minor, are nowhere
dense. Furthermore, under certain closure conditions, nowhere denseness constitutes the frontier of
parameterized tractability for natural classes of problems. 
For instance, while the first-order model-checking problem is fixed-parameter tractable on every nowhere dense class $\Cc$~\cite{GroheKS17},
on every subgraph-closed class $\Dd$ that is not nowhere
dense, it is as hard as on the class of all graphs~\cite{dawar2009parameterized,dvovrak2013testing}.  Similar
 lower bounds are known for many individual problems,
e.g.~for the \rIndSet problem and the \rDomSet problem, on
subgraph-closed classes
which are not nowhere dense~\cite{drange2016kernelization, PS18}.

%\vspace{-0.3pt}
Towards the goal of
extending the border of algorithmic tractability for the above
mentioned problems beyond graph classes that are closed under taking
subgraphs, we study a very simple and generic algorithmic paradigm
that we call \emph{progressive exploration}.
%To illustrate the
%paradigm, we consider as an example the \domSet problem. Recall that a 
%dominating set in a graph $G$ is a set $D$ such that every $v\in V(G)\setminus D$ has a neighbor in $D$.  The
%algorithmic task is to determine whether an input graph~$G$ admits a
%dominating set of size at least $k$, where $k$ is a given parameter.

The idea of progressive exploration can be applied to a {\em{parameterized subset problem}}: given a graph $G$ and parameter $k$,
we look for a vertex subset $S$ of size $k$ that has some property.
A progressive exploration algorithm works in rounds $i=1,2,\ldots$, where each round $i$ finishes with constructing a {\em{candidate solution}}~$S_i$, short, a \emph{candidate}, and, in case $S_i$ is not a solution, a {\em{witness}} $W_i$ for this fact.
More precisely, in round~$i$, we first attempt to find a candidate $S_i$ of size $k$ which {\em{agrees}} with all the previously found witnesses $W_1,W_2,\ldots,W_{i-1}$,
in the sense that none of them witnesses that $S_i$ is not a feasible solution.
Obviously, if no such $S_i$ exists, then we can terminate and return that there is no solution.
On the other hand, if the found candidate $S_i$ is in fact a feasible solution, then we can also terminate and return it.
Otherwise, we find another witness~$W_i$, which witnesses that all the candidates found so far -- $S_1,\ldots,S_{i-1},S_i$ -- are not feasible solutions,
and we proceed to the next round.
In this way, we progressively explore the whole solution
space, while constructing more and more problematic witnesses that the future candidates
must agree with, until we either find a solution or enough witnesses to certify that no solution exists.

\begin{comment}
In round $i$ of the algorithm, we consider a vertex set
$S_i$ of size $k$ as a candidate solution (starting with an arbitrary
set in the first step).  We either verify that $S_i$ is indeed a
solution and output this solution, or we find a set $W_i$ highlighting
a conflict with all candidate solutions considered so far. 
%For
%example, in the case of \domSet, such a conflict set $W_i$ could be 
%single vertex $v$ such that every candidate solution $S_j$ does not
%dominate $v$. Of course, for different algorithmic problems, 
%different forms of conflicts must be highlighted. 
Let us assume for now that for the problem under consideration 
we can always find a small conflict set $W_i$ efficiently.  We then
attempt to find a new candidate solution $S_{i+1}$ which is not
conflicting~$W_i$. If we fail to do so, we output $W_i$ as a witness
for the fact that no solution of size
$k$ exists in $G$. Otherwise, we proceed to the next round with the candidate
set $S_{i+1}$, which again we assume to be efficiently computable for
the moment. In this way, we progressively explore the whole solution
space, while constructing more and more conflict sets that a solution
must obey, until we either find a solution or a conflict set
witnessing that no solution exists.
\end{comment}

Progressive graph exploration is a generic 
approach to solving graph problems, which so far rather resembles a wishful-thinking heuristic than a viable algorithmic methodology.
Such algorithms can be applied to any input graph,
however, a 
priori there are multiple problem-dependent details to be filled.
First, in order to implement the iteration, we need to efficiently
compute candidates $S_i$ and small witnesses $W_i$ in every round.
Second, to analyze the running time we need to give an upper bound on the number of rounds in which the algorithm
terminates. If we can guarantee that each round can be
implemented efficiently and that the number of rounds is bounded in
the parameter $k$, we immediately obtain a fixed-parameter algorithm
for the problem under consideration. 

In this work we study properties
of algorithmic problems and graphs that ensure these desired features.
We consider problems for which the property that a candidate $S$ agrees with a witness $W$
is, in a way to be made precise, definable by a first-order formula. 
We identify model-theoretic properties of formulas that
lead to efficient progressive exploration algorithms. The properties  
that guarantee the existence of small witness sets are 
variants of the \emph{Helly property}, called 
\emph{nfcp} in model theory. 
The property that guarantees that the progressive
exploration algorithms stop after a bounded number of rounds is
the model-theoretic notion of \emph{stability}.
Under these conditions, for problems formulated using short distances in the graph, 
we are able to implement progressive exploration efficiently,
yielding fast and simple fixed-parameter algorithms.
See the caption in \cref{fig:overview} for an
overview of the paper.

\begin{figure}\centering\vspace{-2em}
    \includegraphics[scale=0.55]{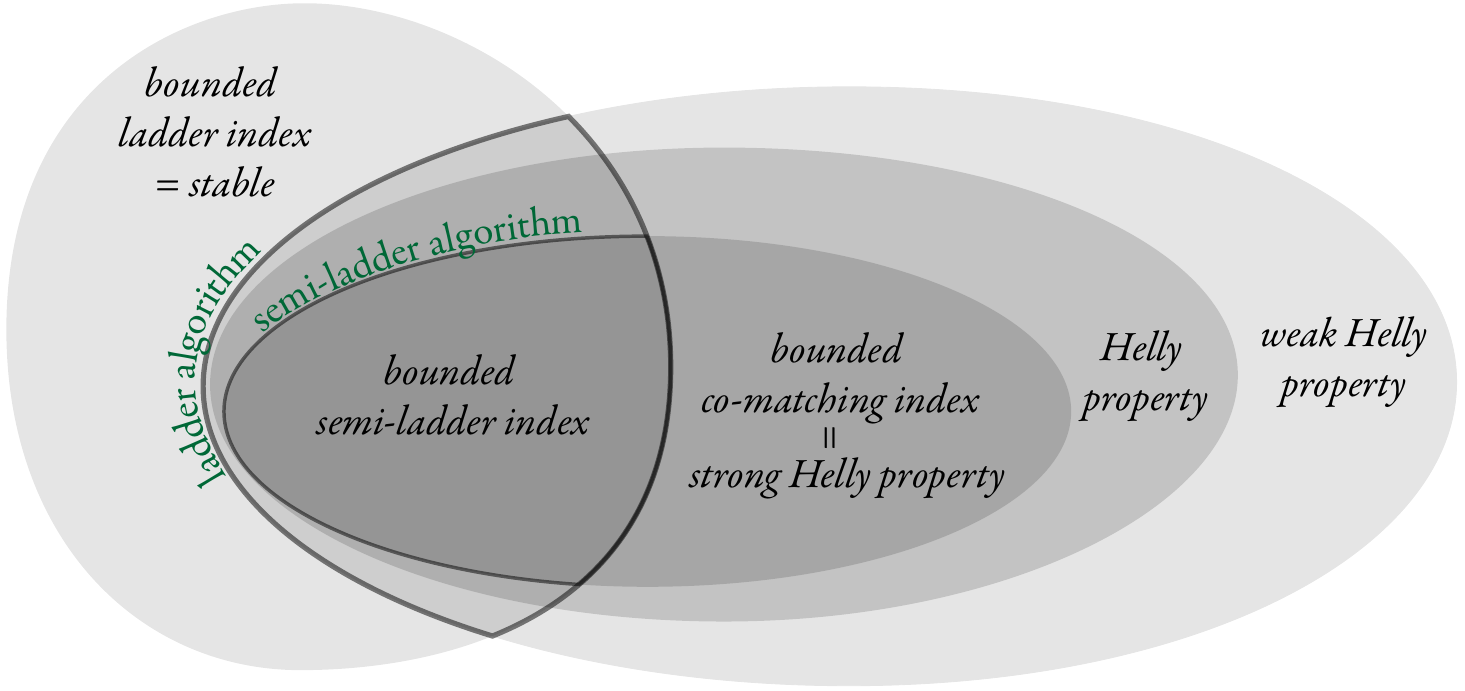}\vspace{-.5em}
    \caption{The figure depicts various properties 
    of classes of bipartite graphs, which are introduced in \cref{sec:helly}.
    Domination- and independence-type problems
    studied in \cref{sec:problems}
    reduce to the problem of determining whether the right part of a given bipartite graph has a common neighbor.
    In \cref{sec:algo} 
      two algorithms for the latter problem are devised, and their domains of applicability are marked above. The  
      \emph{ladder algorithm} has a larger domain, but requires a more powerful access oracle and has higher running time.
      Finally, we apply these algorithms to specific graph classes, yielding new fixed-parameter tractability results for domination- and independence-type problems.\vspace{-1em}
    }\label{fig:overview}
\end{figure}

We demonstrate our approach by applying it to the \rDomSet problem and 
the \rIndSet problem on a variety of restricted graph classes.
Precisely, we prove that:
\begin{itemize}
\item For every $r\in \N$ and graph class $\Cc$ that is either nowhere dense, or is a power of a nowhere dense class, or is the class of map graphs, 
the \rDomSet problem on any graph $G\in \Cc$ can be solved in time $2^{\Oh(k\log k)}\cdot \|G\|$.
Here and throughout the paper, $\|G\|$ denotes the total number of vertices and edges in a graph $G$, while $|G|$ denotes the number of vertices in $G$.
\end{itemize}
\pagebreak
\begin{itemize}
\item For every $t\in \N$, the \domSet problem on any $K_{t,t}$-free graph $G$ can be solved in time $2^{\Oh(k\log k)}\cdot \|G\|$; here, a graph is {\em{$K_{t,t}$}}-free if it does not contain the complete
bipartite graph $K_{t,t}$ as a subgraph.
\item For every $r\in \N$ and nowhere dense graph class $\Cc$, the \rIndSet problem on any graph $G\in \Cc$ can be solved in time $f(k)\cdot \|G\|$, for some function $f$.
\end{itemize}
Actually, for the last result, we also give a different algorithm with running time $2^{\Oh(k\log k)}\cdot \|G\|$, which however does not rely on the concept of progressive exploration and 
uses some black-boxes from the theory of nowhere dense graphs.

Our results extend the limits of tractability for \rDomSet and \rIndSet and, in some cases, improve the best known running times.
We include a comprehensive comparison with the existing literature at the end of \cref{sec:algo}.
However, let us stress here a key point: all our algorithms are derived in a generic way using the idea of progressive exploration, hence they are very easy to implement and
they do not use any algorithmic black-boxes that depend on the class from which the input is drawn. In fact, properties of the considered classes are used only when analyzing the running~time.

\paragraph*{Acknowledgment.}
We thank Ehud Hrushovski for pointing us to the
nfcp property.

\section{Complexity-measures for bipartite graphs}\label{sec:helly}
In this section, we define the basic notions used in this paper, related to various complexity measures
associated with bipartite graphs.

A \emph{bipartite graph} is a triple $G=(L,R,E)$,
where $L$ and $R$ are two sets of vertices, called  the \emph{left part} and \emph{right part}, respectively, and $E\subset L\times R$ is a binary relation whose elements are called \emph{edges}. 
Hence, bipartite graphs with parts $L,R$ correspond bijectively to binary relations
with domain $L$ and codomain~$R$.
Note that each bipartite graph has a uniquely determined left and right part.
Also, those parts are not necessarily~disjoint.

\paragraph*{Ladders, semi-ladders, and co-matchings.}
We now define various complexity measures for bipartite graphs, based on the size of a largest ``obstruction'' found in a given bipartite
graph. There are several types of obstructions, leading to different complexity measures.
We start with defining the various types of obstructions.
Let $G=(L,R,E)$ be a bipartite graph. Two sequences, $a_1,\ldots,a_n\in L$ and $b_1,\ldots,b_n\in R$, form: 
\begin{itemize}
\item a {\em{co-matching}} of order $n$ in $G$ if we have $(a_i,b_j)\in E \Longleftrightarrow i\neq j$, for all $i,j\in \{1,\ldots,n\}$;
\item a {\em{ladder}} of order $n$ in $G$ if we have $(a_i,b_j)\in E \Longleftrightarrow i>j$, for all $i,j\in \{1,\ldots,n\}$; and
\item a {\em{semi-ladder}} of order $n$ in $G$ if $(a_i,b_j)\in E$ for all $i,j\in \{1,\ldots,n\}$ with $i>j$, and $(a_i,b_i)\notin E$ for all $i\in \{1,\ldots,n\}$.
\end{itemize}
Note that in case of a semi-ladder we do not impose any condition for $i<j$.
Observe that any ladder of order $n$ and any co-matching of order $n$ are also semi-ladders of order $n$.

\begin{figure}[h]
                \centering
                \def\svgwidth{0.8\columnwidth}
                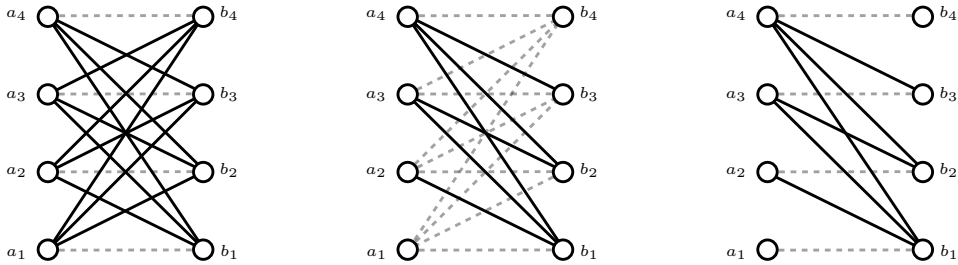
\caption{A co-matching, a ladder, and a semi-ladder of order $4$, respectively. Dashed gray lines represent non-edges.}\label{fig:obstructions}
\end{figure}

The \emph{co-matching index} of a bipartite graph
is the maximum order of a co-matching that it contains.
A class of bipartite graphs has \emph{bounded}
co-matching index if the supremum of the co-matching indices of its members is finite.
We define analogous notions for the \emph{ladder index} and the \emph{semi-ladder index}, in the expected way.

\smallskip

In this paper, we will often not care about the precise bounds on the indices of graphs, 
and it will only matter whether the respective index is bounded in a given class. 
Classes of bipartite graphs with bounded ladder index are also known 
as \emph{stable} classes.
We will later relate the property of having
a bounded co-matching index to a variant of the \emph{Helly property}.
However, using a simple Ramsey argument,
we now observe that boundedness of the semi-ladder index is equivalent to  boundedness of both the co-matching and the ladder index.
Let us first state Ramsey's theorem in the form used in this paper.

\begin{theorem}[Ramsey's theorem]\label{thm:ramsey}
 For all $c,\ell\in\N$ there exists a number $R^c(\ell)$ with the following property.
 If the edges of a complete graph on $R^c(\ell)$ 
 vertices are colored using $c$ colors,
 then there is a set of $\ell$ vertices which is \emph{monochromatic},
 that is, all edges with both endpoints in this set are of the same color.
\end{theorem}

The standard proof of Ramsey's theorem yields an upper bound of $R^c(\ell)\leq c^{c\ell-1}$ for $c\geq 2$. We include a proof
for completeness. However, we defer the proof, and
other proofs which are merely technical, to \cref{sec:omitted} in order
to not disturb the flow of presentation.

From now on, we adopt the notation $R^c(\ell)$ for the multicolored Ramsey numbers as described in \cref{thm:ramsey}.
Also the proof of the following lemma is relegated to \cref{sec:omitted}.

%We can now prove that the boundedness of the semi-ladder index
%is equivalent to boundedness of both the co-matching 
%and the ladder index. 

\begin{lemma}\label{lem:comatch&semiladder}
    A class of bipartite graphs has finite semi-ladder index if and only if 
    both its ladder index and its co-matching index are finite.
\end{lemma}
%\begin{proof}
%The left-to-right implication is immediate:
%every ladder of order $n$ and every co-matching of order $n$
%are also semi-ladders of order $n$, so the semi-ladder index is always an upper bound on both the co-matching and the ladder index.
%
%For the right-to-left implication, we prove that if a bipartite graph $G=(L,R,E)$ has both the ladder and the co-matching index smaller than some $\ell\in \N$, then its semi-ladder index is smaller than $q=R^2(\ell)$.
%Suppose for contradiction that $a_1,\ldots,a_q\in L$ and $b_1,\ldots,b_q\in R$ form a semi-ladder of order $q$ in $G$.
%Color all pairs $(i,j)$ with $1\le i<j\ge q$ red or blue, depending on whether $(a_i,b_j)\in E$ or not.
%By Ramsey's theorem, there is a subset of $\set{1,\ldots,q}$
%of size $\ell$ which is monochromatic, i.e., which only spans 
%red edges, or only spans blue edges.
%In the first case, the corresponding vertices $a_i$ and $b_i$ form a co-matching of order $\ell$ in $G$,
%and in the latter case we analogously exhibit a ladder of order $\ell$ in $G$.
%In any case, this is a contradiction.
%\end{proof}

%\medskip
\paragraph*{Helly property and its variants.}\label{sec:helly-variants}
Let $p\in \N$ and let $G=(L,R,E)$ be a bipartite graph.
We say that a subset $B\subseteq R$ is {\em{covered}} by a subset $A\subseteq L$ if there exists a vertex $a\in A$ which is adjacent to all the vertices of $B$.
Then subsets $A$ and $B$ have the {\em{$p$-Helly property}} 
if either $B$ is covered by $A$, or $B$ contains a subset of size at most $p$ that is not covered by $A$.
We shall say that $G$ has:
\begin{itemize}
\item  the {\em{weak $p$-Helly property}} if $L$ and $R$ have the $p$-Helly property;
\item  the {\em{$p$-Helly property}} if $L$ and $B$ have the $p$-Helly property, for all $B\subseteq R$; and
\item  the {\em{strong $p$-Helly property}} if all $A\subset L$ and $B\subset R$ have the $p$-Helly property.
\end{itemize}
% Obviously, the strong $p$-Helly property implies the $p$-Helly property,
% which implies the weak $p$-Helly property.
We say that a class $\Cc$ of bipartite graphs 
has the (weak/strong) Helly property if there is
some $p\in \N$ such that all graphs in $\Cc$ have the (weak/strong) $p$-Helly property. The Helly property is called
\emph{nfcp} in model theory~\cite[Chapter II.4]{shelah1990classification}.

It turns out that the strong $p$-Helly property corresponds precisely to
having co-matching index at most~$p$. A proof of the lemma is presented in \cref{sec:omitted}.

\begin{lemma}\label{lem:hel-com}
    Let $p\in\N$ and $G$ be a bipartite  graph.
    Then $G$ has the strong $p$-Helly property if and only if it has co-matching index at most $p$.
\end{lemma}
%\begin{proof}
%    Let $G=(L,R,E)$.
%    For the left-to-right implication,
%    suppose $a_1,\ldots,a_q\in L$ and $b_1,\ldots,b_q\in R$ form a co-matching of order $q$ in $G$, for some $q\in \N$.
%    Then $A=\{a_1,\ldots,a_q\}$ and $B=\{b_1,\ldots,b_q\}$ do not have the $q$-Helly property, implying that $q\leq p$.
%    This means that the co-matching index of $G$ is at most $p$.
%
%For the right-to-left implication, 
%it is enough to show that if $G$ has co-matching index at most $p$, then it has the weak $p$-Helly property.
%Indeed, from this it follows that in fact $G$ must have the strong $p$-Helly property, since we can apply the same argument to every induced subgraph of $G$, 
%which also has co-matching index at most $p$.
%
%Thus, assume that $G$ has co-matching index at most $p$ and $R$ is not covered (otherwise we are done).
%Let $B\subset R$ any inclusion-minimal set which is not covered by $L$;
%we show that $|B|\le p$.
%By minimality, for every $b\in B$ there is some $a_b\in L$ 
%which is adjacent to all vertices in $B\setminus \set b$ and not to $b$.
%This means that~$B$ together with $A=\{a_b\colon b\in B\}$ form a co-matching of order $|B|$ in $G$,
%which implies that $|B|\le p$ by the assumption on the co-matching index of $G$.
%\end{proof}

In the following paragraphs we will see 
specific examples of classes of bipartite 
graphs satisfying variants of the (weak/strong) Helly property.

\paragraph*{Bipartite graphs defined by formulas.}
We construct bipartite graphs using logical formulas.
In principle, we could consider formulas of any logic,
but in this paper we only consider first-order logic in the vocabulary of graphs,
i.e., using the binary relation symbol $E$ representing 
edges, the binary relation symbol $=$ representing equality,
and logical constructs $\lor,\land,\neg,\rightarrow,\forall,\exists$.
E.g, the property $\dist(x,y)\le 5$, expressing that 
$x$ and $y$ are at distance at most~$5$ in a graph $G$,
can be expressed by a first-order formula using four existential~quantifiers.

We write $\tup x$ to represent a non-repeating tuple of variables.
If $V$ is a set, then $V^{\tup x}$ denotes the set of 
all assignments mapping variables in $\tup x$ to $V$.
Let $\phi(\tup x;\tup y)$ be a formula 
with free variables partitioned into two disjoint tuples, $\tup x$ and $\tup y$.
Given any graph $G$ with vertex set $V$, the formula $\phi$ 
induces a bipartite graph~$\phi(G)$ with
left part $V^{\tup x}$, right part $V^{\tup y}$,
and where $\bar a\in V^{\tup x}$ and $\bar b\in V^{\tup y}$ are 
adjacent if and only if~$\phi(\bar a;\bar b)$ holds in $G$.
If $\Cc$ is a class of graphs
and $\phi(\tup x;\tup y)$ is a formula,
then by $\phi(\Cc)$ we denote
the class of all bipartite graphs $\phi(G)$, for $G\in \Cc$.
We say that $\phi$ has \emph{bounded ladder index} on a class $\Cc$
if the class $\phi(\Cc)$ has bounded ladder index; similarly for the co-matching and the semi-ladder index.
The same definitions apply if instead of graphs we consider
logical structures over some fixed signature, and $\phi(\tup x;\tup y)$ is a formula over that signature. For simplicity,
we consider only graphs in this paper.

We note that the various indices are preserved by adding 
spurious free variables to formulas. Precisely, 
let $\phi(\tup x;\tup y)$ be a first-order formula
and let $\phi'(\tup x';\tup y')$ be the same formula, 
but having extra free variables, i.e., $\tup x$ is a subtuple of $\tup x'$ and $\tup y$ is a subtuple of~$\tup y'$.
Then, for any graph $G$,
the ladder index of $\phi(G)$ is equal to the ladder index of $\phi'(G)$, although the bipartite graphs 
$\phi(G)$ and $\phi'(G)$ may differ.
The same holds for all  other properties studied in this paper:
co-matching index, semi-ladder index, (weak/strong) Helly~property.

%We remark that the same result holds if $\Cc$ is a class 
%of logical structures over some fixed signature.
%The proof relies on Ramsey's theorem, which we state below.

% \subparagraph*{(Semi)-ladder indices and positive boolean combinations.}
The next lemma shows that a positive boolean combination
of formulas with bounded semi-ladder index also has bounded semi-ladder index. 
The proof uses a Ramsey~argument and is presented in 
\cref{sec:omitted}. 

\begin{lemma}\label{lem:bool-comb}
    Let $\phi_1(\tup x;\tup y),\ldots,\phi_k(\tup x;\tup y)$ 
    be formulas and let $\psi(\tup x;\tup y)$ be a positive boolean combination
    of $\phi_1,\ldots,\phi_k$.
    Suppose $G$ is a graph such that $\phi_1(G),\ldots,\phi_k(G)$ have semi-ladder index
 smaller than $\ell$. Then~$\psi(G)$ has semi-ladder index smaller than $R^k(\ell)$. 
\end{lemma}
%\begin{proof}
%    Let $G$ be a graph with vertex set $V$ and suppose that 
%    $\psi(G)$ has semi-ladder index at least $R^k(\ell)$.    
%    Then there are tuples $\tup a_1,\ldots,\tup a_q\in V^{\tup x}$ 
%    and $\tup b_1,\ldots, \tup b_q\in V^{\tup y}$,
%    where $q=R^k(\ell)$,
%    such that for all $i,j\in \set{1,\ldots,q}$,
%    if $i>j$ then $\psi(\tup a_i;\tup b_j)$ holds in $G$,
%    and if $i=j$, then $\psi(\bar a_i;\bar b_i)$ does not hold in $G$.  
%
%\pagebreak
%    Since $\psi$ is a positive 
%    boolean combination of $\phi_1,\ldots,\phi_k$,
%    whenever  $(\tup a,\tup b),(\tup a',\tup b')\in V^{\tup x}\times V^{\tup y}$ 
%    are two pairs of tuples, one satisfying $\psi$ in $G$ and 
%    the other satisfying $\neg \psi$ in $G$,
%    then there must be some $p\in \set{1,\ldots,k}$
%    such that $(\tup a,\tup b)$ satisfies $\phi_p$ and
%    $(\tup a',\tup b')$ satisfies $\neg \phi_p$.
%    
%
%Hence, for all $q\ge i>j\ge 1$,
%we may color the pair $(i,j)$ with some number 
%$p\in \set{1,\ldots,k}$ such that 
%$\phi_p(\tup a_i;\tup b_j)$ holds in $G$ 
%and $\phi_p(\tup a_i;\tup b_i)$ does not hold in $G$.
%
%By Ramsey's theorem, there 
%is a set $X\subset \set{1,\ldots,q}$
%of size $\ell$, and a color $p\in \set{1,\ldots,k}$,
%such that each pair $(i,j)\in X^2$ with $i>j$ 
%is colored with color $p$. It follows that
%$\phi_p(a_i;b_j)$ holds in $G$ for all $i,j\in X$ with $i>j$, and $\varphi_p(\tup a_i;\tup b_i)$ does not hold in $G$ for all $i\in X$.
% This shows that $\phi_p(G)$ contains a semi-ladder of order $\ell$,
% contrary to the assumption.
%\end{proof}

We remark 
that the property of 
having bounded ladder index is 
preserved by taking arbitrary boolean combinations,
not just positive ones.
%Finally, the analogue of Lemma~\ref{lem:bool-comb}
%fails for the co-matching index if positive boolean combinations are considered, but still holds if we restrict attention
%to conjunctions of atomic formulas.

We will later need the following variant of \cref{lem:bool-comb}, which provides a sharper bound for formulas of a special form.
Again, the proof relies on a Ramsey-like argument and is presented in \cref{sec:omitted}. 

%\vspace{-1pt}
\begin{lemma}\label{lem:disj-comb}
    Let $\psi(\tup x;\tup y)=\bigvee_{j=1}^k \phi(\tup x^j;\tup y)$ for some $k\geq 2$, 
    where $\phi(\tup x;\tup y)$ is a formula and $\tup x^1,\ldots,\tup x^k$ are permutations of $\tup x$.
    Suppose $G$ is a  graph such that $\phi(G)$ has semi-ladder index
 smaller than $G$. Then 
    $\psi(G)$ has semi-ladder index smaller than~$k^{\ell-1}$. 
\end{lemma}
\paragraph*{Stability.}
The  classes of bipartite graphs with bounded ladder 
index that are relevant to this paper are provided by the following result, due to Podewski and Ziegler~\cite{podewski1978stable}. %Their proof uses
%model-theoretic methods, exploiting the connection with stability theory.

%\vspace{-1pt}
\begin{theorem}[\cite{podewski1978stable}, cf.~\cite{adler2014interpreting, PilipczukST17}]\label{thm:stability}
    Let $\Cc$ be a nowhere-dense class of graphs and let $\phi(\tup x;\tup y)$ be any first-order formula.
    Then $\phi$ has bounded ladder index on $\Cc$.
\end{theorem}

%\vspace{-1pt}
The above result 
was originally stated for \emph{superflat graphs}.
The connection with nowhere-denseness was observed in~\cite{adler2014interpreting},
and a proof providing explicit bounds was given in~\cite{PilipczukST17}. The observation of~\cite{adler2014interpreting} is the starting point of this work, as it brings to light the connection
between model theory and computer science.
%which 
%is further studied in this paper.

\begin{comment}
\smallskip
We remark that many examples of classes of bounded ladder index arise in stability theory, whose focus is mainly on infinite structures. A logical structure $\mathbb A$ is \emph{stable} if for every first-order formula $\phi(\tup x;\tup y)$,
the bipartite graph $\phi({\mathbb A})$ has bounded ladder index.
For example, the field of complex numbers $(\mathbb C,+,
\times)$ is stable, and so is any infinite abelian group.
Similarly, if $\cal A$ is the class of all  abelian groups
and $\phi(\tup x;\tup y)$ is any first-order formula in the language of groups, then $\phi$ has bounded ladder index on $\cal A$.
\end{comment}

\section{Domination and independence problems}\label{sec:problems}

We consider subset problems, where in a given graph we look for a solution $S$ of size $k$ that satisfies some property, whose dissatisfaction can be witnessed by a small subset of vertices~$W$.
Moreover, checking a candidate solution $S$ against a witness $W$ can be expressed in first-order logic.
Thus, a problem of interest can be expressed by a sentence of the form $\exists_{\tup x}\forall_{\tup y}\phi(\tup x;\tup y)$, 
for a suitable formula $\phi(\tup x;\tup y)$, where $\tup x$ is a tuple of $k$ variables that represent a candidate $S$, while $\tup y$ is a tuple of $\ell$ variables that represent a witness $W$.

\begin{example}\label{ex:problems}
The \rDomSet problem for parameter $k$ can be expressed as above using the formula $\delta^k_r(\tup x;y)=\bigvee_{i=1}^k \delta_r(x_i,y)$, where $\delta_r(x,y)$ is a formula that checks whether $\dist(x,y)\leq r$,
and $\tup x=(x_1,\ldots,x_k)$. Similarly, the \rIndSet problem for parameter $k$ can be expressed using the formula $\eta^k_r(\tup x;y)=\bigwedge_{1\leq i<j\leq k} \eta_r(x_i,x_j,y)$, where $\eta_r(x,x',y)$
is a formula that checks whether $\dist(x,y)+\dist(x',y)>r$.
\end{example}

Observe that a graph $G$ satisfies the sentence $\exists_{\tup x}\forall_{\tup y}\phi(\tup x;\tup y)$ if and only
if the right part of the bipartite graph $H=\phi(G)$ is covered by the left side,  i.e., all vertices in the right part (witnesses) have a common neighbor (solution).
We call this abstract problem --- checking whether the right part of a given bipartite graph $H$ is covered by the left side --- the {\sc{Coverage}} problem.

%Thus, the initial problem on $G$
%reduces to the problem of deciding whether the right part of a given bipartite graph $H=\phi(G)$
%is covered by the left side. We call this the {\sc{Coverage}} problem.

Note that the size of the bipartite graph $\phi(G)$ is polynomial in the size of $G$,
where the exponent depends on the number of free variables in $\phi$, which is usually the parameter we are interested in.
As we are aiming at fixed-parameter algorithms, we cannot afford to even construct the whole bipartite
graph~$\phi(G)$. Therefore, we will design algorithms that solve the {\sc{Coverage}} problem using an oracle access 
to the bipartite graph graph $H=\phi(G)$, where the oracle calls will be implemented using subroutines on the original graph $G$.
The running time of these algorithms, expressed in terms of the number of oracle calls, will be bounded only in terms of quantities (ladder indices, numbers governing Helly property, etc.)
related to the class of bipartite graphs $\phi(\Cc)$, where $\Cc$ is the considered class of input graphs.

\pagebreak
 Therefore, to obtain 
 an algorithm for solving the initial problem on a given graph class~$\Cc$ we proceed in two steps:
\begin{itemize}
    \item Prove that the class $\phi(\Cc)$ has a suitable Helly-type property
    and bounded ladder index.
    
    \item Design an algorithm
    for {\sc Coverage}, for input bipartite graphs with suitable Helly-type properties
    and bounded ladder index, that uses only a bounded number of oracle calls.
\end{itemize}

In \cref{sec:algo} we give two such algorithms 
solving {\sc{Coverage}}: 
the {\em{Semi-ladder Algorithm}}, and the {\em{Ladder Algorithm}}. 
The Semi-ladder Algorithm
requires that $H$ has bounded semi-ladder index, 
whereas the Ladder Algorithm requires that $H$
has bounded ladder index and the weak $p$-Helly property, for some fixed $p$. 
Note that by \cref{lem:comatch&semiladder} and~\ref{lem:hel-com}, boundedness of the semi-ladder index is equivalent to boundedness of the ladder index and having the strong $p$-Helly property, for some fixed $p$,
so the prerequisites for the Semi-ladder Algorithm are stronger than for the Ladder Algorithm.
See \cref{fig:overview} for an overview.

We postpone the discussion of the algorithms to \cref{sec:algo},
and for now we focus on exhibiting the suitable properties for various classes of bipartite graphs.
Slightly more precisely, we prove that on certain graph classes, formulas corresponding to 
domination-type problems have bounded semi-ladder index, while
those corresponding to independence-type problems have the weak Helly property and bounded ladder index. 
Hence, in the first case we will apply the Semi-ladder Algorithm, and in the second --- the Ladder Algorithm. 

\paragraph*{Distance formulas and domination-type problems.}
We shall prove fixed-parameter tractability results not only for distance-$r$ domination, but for a more general class of domination-type problems.
Those can be expressed by suitable formulas, as explained next.

%\begin{definition}
    For $r\in \N$, let $\delta_r(x,y)$ be the formula checking whether $\dist(x,y)\leq r$.
    A~\emph{distance formula}
    is a formula~$\phi(\tup x;\tup y)$ which is 
    a  boolean combination
      of {\em{atoms}} of the form $\delta_r(x,y)$, 
      where the variable $x$ occurs in  $\tup x$,
      the variable $y$ occurs in $\tup y$, and $r\in\N$ is any number. 
      The \emph{radius} of a distance formula is the maximal
      number $r$ occurring in its atoms, whereas its 
       \emph{size} is the number of atoms occurring in it.
      A distance formula is 
      \emph{positive} if it is a positive boolean combination of atoms.

      A \emph{domination-type property}
    is a sentence $\psi$ of the form  $\exists_{\tup x}\forall_{\tup y}\,\phi(\tup x;\tup y)$, where $\phi$ is a positive distance formula.  A \emph{domination-type problem}
    is the computational problem of determining whether
    a given graph $G$ satisfies a given domination-type property.
%\end{definition}

\begin{example}\label{ex:dominations}
    Fix $r\in\N$ and let $\tup x=(x_1,\ldots,x_k)$ be a $k$-tuple of variable. 
    Then the formula $\delta^k_r(\tup x;y)$ considered in \cref{ex:problems}
    is a positive distance formula, hence the problem defined by the domination-type property $\exists_{\tup x}\forall_{\tup y}\,\delta^k_r(\tup x;y)$ (aka \rDomSet) is a domination-type problem.
    Similarly, formulas $\phi(\tup x;y)$ expressing the following properties give raise to natural domination-type problems:
    \begin{itemize}
        \item $y$ is at distance at most $r$ from at least two of the vertices $x_1,\ldots,x_k$; and
        \item the sum $\dist(x_1,y)+\dist(x_2,y)+\ldots+\dist(x_k,y)$ is at most $r$.
    \end{itemize}
    On the other hand, the formula $\eta^k_r(\tup x;y)$ considered in \cref{ex:problems} is a distance formula, but it is not positive, and hence it does not yield a domination-type property.
    \end{example}

    From \cref{lem:bool-comb} and~\ref{lem:disj-comb} and the remark about spurious variables not affecting the semi-ladder index, we immediately obtain the following.
    
    \begin{corollary}\label{cor:dist-bound}
        Let $\phi(\tup x;\tup y)$ be a positive distance formula
        of radius $r$ and size $s$.
        If $G$ is a graph such that the semi-ladder index
        of $\delta_q(G)$ is smaller than $\ell$ for all $q\le r$,
        then the semi-ladder index of $\phi(G)$ is smaller than~$R^s(\ell)$. Moreover, if $\phi=\delta_r^k$ as defined in \cref{ex:problems} and $k\geq 2$,
        then the semi-ladder index of $\phi(G)$ is smaller than $k^{\ell-1}$.
    \end{corollary}

% On the other hand, the semi-ladder index can be characterized in
% terms of the neighborhood depth, defined as follows.

% Fix a bipartite graph $G=(L,R,E)$,
% the \emph{neighborhood} of a vertex $v\in L$ is the 
% set $N(v)\subset R$ of vertices $w\in R$ that are adjacent to $v$.
% The \emph{neighborhood depth} of $G$
% is the maximal number $d$ such that 
% there are vertices $v_1,\ldots,v_{d-1}\in L$ 
% for which the following chain is strictly decreasing:
%  $$V(G)\supseteq N(v_1)\supseteq N(v_1)\cap N(v_2)\supseteq
% N(v_1)\cap N(v_2)\cap N(v_3)\supseteq\ldots\supseteq N(v_1)\cap\cdots \cap N(v_{d-1}).$$

% \begin{lemma}\label{lem:nbd}
%     The neighborhood depth of a bipartite graph  is 
%     equal to is semi-ladder index.
% \end{lemma}

% For a (usual) graph $G$, let  $\delta^r(x;y)$ be a formula expressing $\dist(x,y)\le r$,
% and let $N^r(v)=\setof{w\in V(G)}{\dist(v,w)\le r)}$ for $v\in V(G)$.

% \begin{corollary}
%     Let $G$ be a graph and let $r\in\N$ be a number.
% The semi-ladder index of $\delta^r(G)$ is equal to the 
% \emph{$r$-neighborhood depth} of $G$,
% i.e., the maximal 
%  number $d$ such that 
% there are vertices $v_1,\ldots,v_{d-1}$ of $G$ 
% for which the following chain is strictly decreasing:
%  $$V(G)\supseteq N^r(v_1)\supseteq N^r(v_1)\cap N^r(v_2)\supseteq
% N^r(v_1)\cap N^r(v_2)\cap N^r(v_3)\supseteq\ldots\supseteq N^r(v_1)\cap\cdots \cap N^r(v_{d-1}).$$
% \end{corollary}

\paragraph*{Domination problems.}
We first consider domination-type problems and prove that they have bounded semi-ladder indices on any nowhere dense class.
This result can actually be extended beyond nowhere denseness: to powers of nowhere dense classes, to map graphs, and to $K_{t,t}$-free graphs for radius $r=1$.
We define the former two concepts next.

For a graph $G$ and $s\in\N$, let 
$G^s$ denote the graph with the same vertex set as $G$,
where two vertices are adjacent if and only if their distance in $G$ is at most $s$. If $\Dd$ is a graph class, then  $\Dd^s$ denotes the class $\setof{G^s}{G\in\Dd}$.
Note that a power of a nowhere dense class is not necessarily nowhere dense, e.g., the square of the class of stars is the class of complete graphs.

A graph $G$ is 
a map graph if one can assign to each vertex of $G$ a closed,
arc-connected region in the plane so that the interiors of regions are
pairwise disjoint and two vertices of $G$ are adjacent if and only
if their regions share at least one point on their boundaries.
Note that map graphs are not necessarily planar and may contain arbitrarily large cliques, 
as four or more regions may share a single point on their boundaries.

The following result will be used in the next section to obtain
fixed-parameter tractability of domination-type problems over graph classes described above.

\begin{theorem}\label{thm:semi-ladder}
    For any $r\in \N$ and nowhere dense graph class $\Cc$, 
    the formula $\delta_r(x,y)$ has bounded semi-ladder index on $\Cc$.
    The same holds also when $\Cc=\Dd^s$ for some nowhere dense class $\Dd$ and $s\in \N$, when 
    $\Cc$ is the class of map graphs,
    and when $r=1$ and $\Cc$ is the class of $K_{t,t}$-free graphs, for any fixed $t\in \N$.
\end{theorem}

We prove \cref{thm:semi-ladder} in \cref{sec:domination}.
For the case when $\Cc$ is nowhere dense we utilize the well-known characterization of nowhere denseness via {\em{uniform quasi-wideness}}~\cite{nevsetvril2011nowhere}.
In a nutshell, if for some $G\in \Cc$ the graph $\delta_r(G)$ has semi-ladder index $\ell$, then in $G$ we have vertices $a_1,\ldots,a_\ell$ and $b_1,\ldots,b_\ell$
such that $\dist(a_i,b_j)\leq r$ for all $i>j$ and $\dist(a_i,b_i)>r$ for all $i$. Then provided $\ell$ is huge, using uniform quasi-wideness
we can find a large subset $A\subseteq \{a_1,\ldots,a_\ell\}$ of vertices that ``communicate'' with each other only through a set $S$ of constant size --- all paths of length at most $2r$ between
vertices of $A$ pass through~$S$. Now the vertices from $A$ have pairwise different distance-$r$ neighborhoods within $\{b_1,\ldots,b_\ell\}$, but only a limited number of possible
interactions with $S$ (measured up to distance $r$). This quickly leads to a contradiction if~$A$ is large enough.
The cases when $\Cc$ is a power of a nowhere dense class and when $\Cc$ is the class of map graphs follow as simple corollaries from the result for nowhere dense classes.
The case when $\Cc$ is the class of $K_{t,t}$-free graphs is a simple observation: a large semi-ladder in $\delta_1(G)$ enforces a large biclique in $G$.

\medskip
We remark that the argument used in the proof of \cref{lem:nd-dom-semiladder} is similar to the reasoning that shows that graphs from a fixed nowhere dense class admit small {\em{distance-$r$ domination cores}}: 
subsets of vertices whose distance-$r$ domination forces distance-$r$ domination of the whole graph.
This property was first proved implicitly by Dawar and Kreutzer in their FPT algorithm for \rDomSet on any nowhere dense class~\cite{DawarK09}, also using uniform quasi-wideness.
We refer the reader to~\cite[Chapter~3, Section~5]{notes} for an explicit exposition. We will discuss the existence of small 
cores in the more general setting given in this 
paper in \cref{sec:cores}. 

\medskip
Having established boundedness of the semi-ladder index of $\delta_r(x,y)$ on a class $\Cc$, we can use \cref{cor:dist-bound} to extend this to any positive distance formula.
Therefore, by \cref{thm:semi-ladder}, \cref{cor:dist-bound}, and \cref{lem:hel-com} we immediately obtain the following.

    \begin{corollary}\label{cor:cores}
       Let $\Cc$ and $r$ be as in \cref{thm:semi-ladder} 
       and let $\phi(\tup x;\tup y)$ be a positive distance formula of radius at most $r$.
       Then the class $\phi(\Cc)$ has bounded semi-ladder index, so in particular it has the strong Helly property.
    \end{corollary}
    
    In fact, \cref{cor:dist-bound} provides a better control of the semi-ladder index of $\phi(\Cc)$ in terms of the semi-ladder index of $\delta_r(\Cc)$ and the size of $\phi$.
    In the next section we will use these more refined bounds for a precise analysis of the running times.
    
    Note that \cref{cor:cores} does not generalize to arbitrary
    first-order formulas. Indeed, if~$\Cc$ is the class of all edgeless 
    graphs and $\phi(x;y)$ is the formula $x\neq y$, then $\phi(\Cc)$
    is the class of all complements of matchings, which does not even have the weak Helly property. 

\paragraph*{Independence problems.}
We now move to the \rIndSet problem: deciding 
whether a given graph contains a distance-$r$ independent set of size $k$. This property
is most naturally expressed using an existential sentence, and not as a sentence of the form $\exists_{\bar x}\forall_{\bar y}\phi(\tup x;\tup y)$.
However, in \cref{ex:problems} we gave a suitable formula $\eta^k_r(\tup x;y)$ that expresses the problem:
the trick is to phrase the property that $x_1,\ldots,x_k$ are pairwise at distance more than $r$ by saying that for every vertex $y$, 
for all $1\leq i<j\leq k$ the sum of distances from $y$ to $x_i$ and $x_j$ is larger than $r$.
Thus, a vertex $y$ that does not satisfy this condition may serve as a witness that a given tuple $\tup x$ does not form a distance-$r$ independent set.

In \cref{sec:independence} we prove the following.
\begin{theorem}\label{thm:ind-weak}
    Let $\Cc$ be a nowhere-dense class of graphs and let $k,r\in\N$.
    Then the class $\eta^k_r(\Cc)$ has 
    the weak $p$-Helly property, for some $p\in\N$ depending on $k$, $r$, and $\Cc$.
\end{theorem}

It is easy to see that for any $k\ge 2$ and $r\ge 1$, the formula $\eta^k_r(x;y)$ does not have the strong Helly property on the class $\Cc$ of edgeless graphs.
Thus, in general we cannot hope for boundedness of the semi-ladder index of $\eta^k_r(\Cc)$ and use the Semi-Ladder Algorithm.

The proof of \cref{thm:ind-weak} is actually very different from the proof of \cref{thm:semi-ladder}, and presents a novel contribution 
of this work.
Instead of uniform quasi-wideness, we use the characterization of nowhere denseness via the {\em{Splitter game}}~\cite{GroheKS17}.
The idea is that in case a graph $G\in \Cc$ does not have a distance-$r$ independent set of size $k$, there is a small witness of this: 
a set $W$ of size bounded in terms of $k$, $r$, and $\Cc$ such that for every vertex subset $S$ of size $k$, some path of length at most $r$ connecting two vertices of $S$ crosses $W$.
This exactly corresponds to the notion of witnessing expressed by $\eta_r^k$.
Such a witness $W$ is constructed recursively along Splitter's strategy tree in the Splitter game in $G$.
We use the condition that $G$ does not have a distance-$r$ independent set of size $k$ to prove that 
we can find a small (in terms of $k,r,\Cc$) set of ``representative'' moves of the Connector. Trimming the strategy tree to those moves
bounds its size in terms of $k,r,\Cc$, yielding the desired upper bound on the witness size.

We remark that our proof of \cref{thm:ind-weak} can actually be turned 
into an algorithm for the \rIndSet problem on any nowhere dense class $\Cc$ with running time of $2^{\Oh(k\log k)}\cdot \|G\|$.
However, this algorithm is much more complicated than the Ladder algorithm that we explain in the next section, 
and in particular it uses some black-box results from the theory of nowhere dense graph classes.
Details can be found at the end of \cref{sec:independence}.

\newcommand{\formulaGraph}[2]{#1(#2)}

\section{Algorithms}\label{sec:algo}

\paragraph*{Model.}
We consider the following model of an algorithmic search for a solution in a bipartite graph representing the search space.
Consider a bipartite graph $H=(L,R,E)$, where the left side $L$ is the set of {\em{candidates}} and the right side $R$ is the set of {\em{witnesses}}.
An edge between a candidate $a\in L$ and a witness $b\in R$ is interpreted as that $a$ and $b$ {\em{agree}}: $b$ agrees that $a$ is a solution.
Expressed in those terms, {\sc{Coverage}} is the problem of finding a {\em{solution}}: a candidate which agrees with all witnesses.
We will use the terminology of candidates, witnesses, solutions, and agreeing as explained above, as this facilitates the understanding of the algorithms for {\sc{Coverage}} in terms of the original
problems.

As we explained, the considered bipartite graph $H$ will typically be of the form $\varphi(G)$ for some formula~$\varphi(\tup x;\tup y)$ expressing the considered problem.
Thus, $H$ shall represent the whole search space, so we allow our algorithms a restricted access to $H$ via the following {\em{oracles}}.

\bigskip

\begin{minipage}{0.96\textwidth}
\noindent{\bf{Candidate Oracle}}: Given a set of witnesses $B\subseteq R$, the oracle either returns a candidate $a\in L$ that agrees with all witnesses of $B$, or concludes that no such candidate exists.

\medskip

\noindent{\bf{Weak Witness Oracle}}: Given a candidate $a\in L$, the oracle either concludes that $a$ is a solution, or returns a witness $b\in R$ that does not agree with $a$.

\medskip

\noindent{\bf{Strong Witness Oracle}}: Given a set of candidates $A\subseteq L$ and a number $p\in \N$, the oracle either finds a set of witnesses $P\subseteq R$ such that $|P|\leq p$ and 
every candidate of $A$ does not agree with some witness from $P$, or concludes that no such set $P$ exists.
\end{minipage}

\bigskip

\noindent Note that the Weak Witness Oracle can be simulated by the Strong Witness Oracle applied to $A=\{a\}$. We now provide the two algorithms for {\sc{Coverage}} announced in \cref{sec:problems}.

%We now provide two abstract algorithms solving the {\sc{Coverage}} problem in the given bipartite graph $H$, designed using the paradigm of progressive exploration. 
%The first one, called the {\em{Semi-ladder Algorithm}}, assumes bounded semi-ladder index, 
%which is equivalent to having bounded ladder index and strong Helly property, but uses only Weak Witness Oracle calls.
%The second one, the {\em{Ladder Algorithm}}, assumes bounded ladder index and weak Helly property, but utilizes Strong Witness Oracle calls.

\paragraph*{Semi-ladder Algorithm.}
The Semi-ladder Algorithm proceeds in a number of rounds, where each round consists of two steps: first the {\em{Candidate Step}}, and then the {\em{Witness Step}}.
Also, the algorithm maintains a set $B$ of witnesses gathered so far, initially set to be empty. The steps are defined as follows:

\bigskip
\begin{minipage}{0.96\textwidth}

\noindent\textbf{Candidate Step:} Apply the Candidate Oracle to find a candidate $a\in L$ that agrees with all the witnesses in $B$.
If no such candidate exists, terminate the algorithm returning that no solution exists. Otherwise, proceed to the Witness~Step.

\medskip

\noindent\textbf{Witness Step:} Apply the Weak Witness Oracle to find a witness $b\in R$ that does not agree with $a$.
If there is no such witness, terminate the algorithm and return $a$ as the solution. Otherwise, add $b$ to $B$ and proceed to the next round.

\end{minipage}

\bigskip

\noindent The correctness of the algorithm is obvious, while the running time can be bounded by the 
immediate observation that if 
the Semi-ladder Algorithm performs $\ell$ full rounds, then 
the candidates $a_1,\ldots,a_\ell\in L$  discovered in consecutive rounds, together with the witnesses  $b_1,\ldots,b_\ell\in R$  added to $B$ in consecutive rounds, form a semi-ladder in $H$.

% \begin{observation}
% Suppose the Semi-ladder Algorithm applied to a graph $H$ terminated after performing $\ell$ full rounds.
% Let $a_1,\ldots,a_\ell\in L$ be the candidates discovered in consecutive rounds, and let $b_1,\ldots,b_\ell\in R$ be the witnesses added to $B$ in consecutive rounds.
% Then $a_1,\ldots,a_\ell$ and $b_1,\ldots,b_\ell$ form a semi-ladder in $H$.
% \end{observation}

\begin{corollary}\label{cor:semi-ladder-algo}
The Semi-ladder Algorithm applied to a graph $H$ with semi-ladder index~$\ell$ terminates after performing at most $\ell$ full rounds.
Consequently, it uses at most $\ell+1$ Candidate Oracle Calls, each involving a set of witnesses $B$ with $|B|\leq \ell$, and at most $\ell$ Weak Witness Oracle Calls.
\end{corollary}

\paragraph*{Ladder algorithm.}
As before, the Ladder Algorithm maintains the set $B$ of witnesses gathered so far, but also the set $A$ of candidates found so far.
The algorithm is also given a parameter $p\in \N$. 
Again, the algorithm proceeds in rounds, each consisting of the Candidate step and the Witness step, with the following description:

\bigskip

\begin{minipage}{0.96\textwidth}

\noindent\textbf{Candidate Step:} Apply the Candidate Oracle to find a candidate $a\in L$ that agrees with all the witnesses in $B$.
If no such candidate exists, terminate the algorithm returning that no solution exists. Otherwise, add $a$ to $A$ and proceed to the Witness step.

\medskip

\noindent\textbf{Witness Step:} Apply the Strong Witness Oracle to set $A$ and parameter $p$, 
yielding either a set of witnesses $P\subseteq R$ such that $|P|\leq p$ and every candidate from $A$ does not agree with some witness from $P$,
or a conclusion that no such set $P$ exists. In the former case, add $P$ to $B$ and proceed to the next round. In the latter case, terminate the algorithm returning that a solution exists.

\end{minipage}

\bigskip

\noindent Note that the algorithm actually never finds a solution, but only may claim its existence in the Witness Step, and this claim is not substantiated by having a concrete solution in hand. 
However, the observation is that assuming the weak $p$-Helly property, the structure discovered by the algorithm is sufficient to deduce the existence of a solution. 

\begin{lemma}\label{lem:correctness-ladder}
The Ladder Algorithm applied with parameter $p$ in a bipartite graph with the weak $p$-Helly property is always correct.
\end{lemma}
\begin{proof}
If the algorithm returns that there is no solution, this claim is justified by finding a set of witnesses that cannot simultaneously agree with any candidate, which implies that no solution exists.
If the algorithm returns that there is a solution, then this is because it constructed a set of candidates $A$ such that for every subset of witnesses $P\subseteq R$ with $|P|\leq p$, some candidate within $A$
agrees with all the witnesses from $P$. If there was no solution, then by the $p$-Helly property there would exist a set $P_0$ of at most $p$ witnesses, for which there would be no candidate simultaneously
agreeing with all the witnesses of $P_0$; this would contradict the previous conclusion. Hence, the algorithm's claim that there exists a solution is correct.
\end{proof}

Finally, we show that if $H$ has ladder index bounded by $\ell$, then the algorithm terminates in a number of rounds bounded in terms of $\ell$ and $p$.
For this we observe that during its execution, the algorithm in fact constructs a ladder in an auxiliary bipartite graph $H'$ with candidates $a$ on the left side and sets of witnesses $P$ on the right side,
and the ladder index of $H'$ can be bounded in terms of $p$ and the ladder index of $H$ using a Ramsey~argument.

\begin{lemma}\label{lem:ladder-conj}
The Ladder Algorithm applied with parameter $p$ to a bipartite graph $H$ with ladder index smaller than $\ell$ terminates after performing less than $R^p(2\ell)$ full rounds.
\end{lemma}
\begin{proof}
Let $H'$ be the bipartite graph constructed as follows: the left side of $H'$ consists of the candidate set $L$, 
the right side of $H'$ consists of the family of all sets of witnesses $P\subseteq R$ satisfying $|P|\leq p$, and a candidate $a\in L$ is adjacent in $H'$ to a set of witnesses $P\subseteq R$ if and only if
$a$ agrees with all witnesses of~$P$. 

We claim that $H'$ has ladder index smaller than $R^p({2\ell})$. Note that this will conclude the proof for the following reason:
if $a_1,\ldots,a_t$ and $P_1,\ldots,P_t$ are candidates and witness sets discovered by the algorithm during consecutive rounds, where $t$ is the number of full rounds performed,
then $a_1,\ldots,a_t$ and $P_1,\ldots,P_t$ form a ladder in $H'$.

For contradiction, suppose that $a_1,\ldots,a_q$ and $P_1,\ldots,P_q$ form a ladder of order $q=R^p(2\ell)$ in~$H'$.
Arbitrarily enumerate the elements of each set $P_i$ as $b_i^1,\ldots,b_i^p$, repeating elements in case $|P_i|<p$.
Now, color every pair $(i,j)$ satisfying $q\ge i>j\ge 1$ with any number $s\in \{1,\ldots,p\}$ such that $a_j$ does not agree with $b_i^s$; such $s$ exists because $a_j$ and $P_i$ are not adjacent in $H'$.
By Ramsey's theorem, we may find an index subset $X\subseteq \{1,\ldots,q\}$ of size $2\ell$ and a color $s\in \{1,\ldots,p\}$ such that all pairs $(i,j)$ with $i>j$ and $i,j\in X$ have color $s$.
Let $c_1,\ldots,c_{2\ell}$ be the elements of $\{a_i\colon i\in X\}$ ordered as in the sequence $a_1,\ldots,a_q$, and let $d_1,\ldots,d_{2\ell}$ be the elements of $\{b_i^s\colon i\in X\}$ ordered
as in the sequence $b_1^s,\ldots,b_q^s$.
Then it follows that $c_1,c_3,\ldots,c_{2\ell-1}$ and $d_2,d_4,\ldots,d_{2\ell}$ form a ladder of order $\ell$ in $H$, a contradiction.
\end{proof}

\begin{corollary}\label{cor:ladder-algo}
The Ladder Algorithm applied with parameter $p$ to a graph $H$ with ladder index smaller than~$\ell$ and the weak $p$-Helly property, 
always returns the correct answer and terminates after performing at most $q=R^p(\ell)-1$ full rounds.
Consequently, it uses at most $q+1$ Candidate Oracle Calls, each involving a set of witnesses $B$ with $|B|\leq pq$, and at most $q$ Strong Witness Oracle Calls, each involving a set of candidates $A$ with $|A|\leq q$.
\end{corollary}

\paragraph*{Implementing the oracles.}
The last missing ingredient for obtaining our algorithmic results is an efficient implementation of the oracles for bipartite graphs of the form $\varphi(G)$, where
$G$ is the input graph and $\varphi(\tup x,\tup y)$ is a formula expressing the considered problem. We describe such an implementation whenever~$\varphi$ is a distance formula.

We use the concept of {\em{distance profiles}} and {\em{distance profile complexity}}.
Let $G$ be a graph and let $S$ be a set of 
its vertices. For a vertex $v$ of $G$, the \emph{distance-$r$ profile} 
of $v$ on $S$, denoted $\profile^{G,S}_r(v)$, is the function mapping $S$ 
to $\{0,1,\ldots,r,\infty\}$ such that for $s\in S$, 
$$\profile^{G,S}_r(v)(s)=\begin{cases}\dist_G(v,s) & \quad \textrm{if }\dist_G(v,s)\leq r,\\ \infty & \quad \textrm{otherwise.}\end{cases} $$
The {\em{distance-$r$ profile complexity}} of $G$ is the function from $\N$ to $\N$ defined as
$$\numprofiles{r}{G}{m} = \max_{S\subseteq V,\, |S|\leq m} |\{\profile^{G,S}_r(u)\colon v\in V(G)\}|.$$
That is, this is the maximum possible number of different functions from $S$ to $\{0,1,\ldots,r,\infty\}$ realized as distance-$r$ profiles on $S$ of vertices of $G$, over all vertex subsets $S$ of size at most~$m$.
For a graph class $\Cc$, we denote $\numprofiles{r}{\Cc}{m}=\sup_{G\in \Cc} \numprofiles{r}{G}{m}$.

Note that for any graph $G$ and $r,m\in \N$ we have $\numprofiles{r}{G}{m}\leq (r+2)^{m}$, as this is the total number of functions from a set of size $m$ to $\{0,1,\ldots,r,\infty\}$.
This bound is exponential in $m$, however it is known that on nowhere dense classes an almost linear bound holds.

\begin{lemma}[\cite{eickmeyer2016neighborhood}]\label{lem:num-profiles}
Let $\Cc$ be a nowhere dense class of graphs. Then for every $r\in \N$ and $\eps>0$ there exists a constant $c_{r,\eps}$ such that
$\numprofiles{r}{\Cc}{m}\leq c_{r,\eps}\cdot m^{1+\eps}$ for all $m\in \N$.
\end{lemma}

We remark that the conclusion of \cref{lem:num-profiles} still holds when $\Cc$ is any fixed power of a nowhere dense class, and when $\Cc$ is the class of map graphs.
Moreover, when $\Cc$ is the class of $K_{t,t}$-free graphs for some $t\in \N$, then $\numprofiles{1}{\Cc}{m}\leq \Oh(m^t)$.
We prove these statements in \cref{lem:prof-comp-powers} and~\ref{lem:prof-comp-ktt} in \cref{sec:omitted}.

\smallskip
We are ready to give implementations for the oracles. The main idea is that because we are working with a distance formula, when looking for, say, a candidate that agrees with all witnesses in a set $B$, the
only information that is relevant about any vertex is its distance-$r$ profile on the set $S$ consisting of all vertices appearing in the tuples of $B$. Hence, there are only $\numprofiles{r}{G}{|S|}$
different ``types'' of vertices, and instead of checking all $k$-tuples of vertices in the graph, we can check all $k$-tuples of types.

\begin{lemma}\label{lem:oracles}
Fix a distance formula $\varphi(\tup x;\tup y)$ of radius $r$ and with $|\tup x|=c$ and $|\tup y|=d$.
Then for an input graph $G=(V,E)$, there are implementations of oracle calls in $\formulaGraph{\varphi}{G}$ that achieve the following running times:
\begin{itemize}
\item Candidate Oracle: time $\Oh(|B|\cdot \|G\|+|B|\cdot\numprofiles{r}{G}{d|B|}^{c})$ for a call to $B\subseteq V^{\tup y}$;
\item Weak Witness Oracle: time $\Oh(\|G\|+\numprofiles{r}{G}{c}^{d})$ for a call to $\tup a\in V^{\tup x}$;
\item Strong Witness Oracle: time $\Oh(|A|\cdot \|G\|+|A|\cdot \numprofiles{r}{G}{c|A|}^{pd})$ for a call to $A\subseteq V^{\tup x}$ and~$p\in \N$.
\end{itemize}
\end{lemma}
\begin{proof}
Consider first the Candidate Oracle.
Let $S\subseteq V$ be the set of all vertices contained in tuples from~$B$; then $|S|\leq d|B|=\Oh(|B|)$.
Apply a BFS from every vertex of $S$ to compute, for every $u\in V$, a vector of distances from $u$ to all the vertices of $S$.
This takes time $\Oh(|S|\cdot \|G\|)=\Oh(|B|\cdot \|G\|)$.
Given this information, for every vertex $u\in V$ we may compute its distance-$r$ profile on $S$, 
and in particular we may compute the set $\Pp$ of all distance-$r$ profiles on $S$ realized in $G$.
This can be done in time $\Oh(|S|\cdot |V|)$ by adding profiles of all the vertices to a trie.
Note that $|\Pp|\leq \numprofiles{r}{G}{d|B|}$.

Next, we look for a candidate $\tup a\in V^{\tup x}$ that agrees with all elements of $B$ by considering all $c$-tuples of profiles in $\Pp$.
Here, the crucial observation is that since $\varphi$ is a distance formula, whether a candidate $\tup a$ agrees with a fixed witness $\tup b\in B$ depends 
only on the distance-$r$ profiles of the entries of $\tup a$ on $S$. Hence, if suffices to consider only $c$-tuples of profiles in $\Pp$, and not all $c$-tuples of vertices of $G$.
Checking a given $c$-tuple of profiles in $\Pp$ against a fixed witness $\tup b\in B$ can be done in constant time, 
so the whole search for a candidate can be executed in time $\Oh(|B|\cdot |\Pp|^c)\leq \Oh(|B|\cdot \numprofiles{r}{G}{d|B|}^c)$.
Note that a suitable candidate can be retrieved by storing together every profile of $\Pp$ any candidate $\tup a\in V^{\tup x}$ realizing this profile.

The implementation of the Strong Witness Oracle follows from applying exactly the same idea: $S$ is the set of all vertices involved in tuples in $A$, there are at most $\numprofiles{r}{G}{c|A|}$
relevant distance-$r$ profiles on $S$, and we need to search through $pd$-tuples of such profile.
The Weak Witness Oracle can be implemented by simply running the Strong Witness Oracle for a singleton set $A$ and $p=1$.
\end{proof}

\paragraph*{Algorithmic consequences.}
We are ready to present our algorithmic corollaries, promised in \cref{sec:intro}.
Throughout this section, when stating parameterized running times we use $k$ to denote the target size of a solution (i.e., distance-$r$ dominating or independent set).
We start with the domination problems.

\begin{theorem}\label{thm:semi-ladder-algo}
Fix $r\in \N$ and let $\Cc$ be a class of graphs such that for each $q\le r$, the class $\delta_q(\Cc)$ has finite semi-ladder index.
Then, for any positive distance formula $\phi(\tup x;\tup y)$
of radius at most $r$ and size $k$, 
the domination-type problem corresponding to $\phi$ can be solved 
 on $\Cc$ in time $f(k)\cdot \|G\|$, for some function $f$.
\end{theorem}    
\begin{proof}
W.l.o.g. we can assume that $|\tup x|,|\tup y|\leq k$.
Let $\ell\in \N$ be such that $\delta_q(\Cc)$ has semi-ladder index smaller than $\ell$, for all $q\le r$.
Given a graph $G$, we apply the Semi-Ladder Algorithm for the {\sc{Coverage}} problem in the graph $\formulaGraph{\phi}{G}$ with implementations of oracles provided by \cref{lem:oracles}.
By \cref{lem:bool-comb} we conclude the semi-ladder index of $\formulaGraph{\phi}{\Cc}$ is bounded by~$R^k(\ell)$.
Now the claimed running time follows immediately from \cref{cor:semi-ladder-algo} and \cref{lem:oracles}.
\end{proof}

\begin{remark}\label{rem:runtime}
% Let us analyze more closely the upper bounds on function $f$ provided by the proof of Theorem~\ref{thm:semi-ladder-algo}.
By \cref{cor:semi-ladder-algo} and \cref{lem:oracles}, the running time is actually
$\Oh(p\cdot \numprofiles{r}{\Cc}{p}^k\cdot \|G\|)$, where~$p$ is the semi-ladder index of $\formulaGraph{\phi}{G}$.
By \cref{lem:bool-comb}, we have that $p\leq R^k(\ell)$, which is upper-bounded by~$k^{k\ell-1}$ for $k\geq 2$ (see \cref{lem:ramsey-precise}).
Combining this with the trivial upper bound $\numprofiles{r}{\Cc}{p}\leq (r+2)^{p}$ yields $f(k)\leq 2^{2^{\Oh(k\log k)}}$, where $r$ and $\ell$ are considered fixed constants.
However, if a priori we know for the graph class $\Cc$ that $\numprofiles{r}{\Cc}{m}$ is polynomial in $m$, instead of exponential,
then by the analysis above we obtain $f(k)\leq 2^{\Oh(k^2\log k)}$.
Finally, by \cref{lem:disj-comb}, for $\phi=\delta^k_r$ --- the formula corresponding to the \rDomSet problem --- we can use a sharper bound of $p\leq k^{\ell-1}$.
Thus, for this case we obtain an upper bound of $f(k)\leq 2^{\poly(k)}$ in the general setting, and $f(k)\leq 2^{\Oh(k\log k)}$ when $\numprofiles{r}{\Cc}{m}$ is polynomial in $m$.
\end{remark}

%Note that the term $k^{2\ell}\cdot \numprofiles{r}{\Cc}{k^{2\ell}}^k$ is always bounded by $2^{\poly(k)}$ for fixed $r$ and $\ell$, because $\numprofiles{r}{\Cc}{k^{2\ell}}\leq (r+2)^{k^{2\ell}}$.
%However, whenever $\numprofiles{r}{\Cc}{m}$ is polynomial in $m$, then this term is bounded by $2^{\Oh(k\log k)}$.
Now, using \cref{thm:semi-ladder-algo} together with combinatorial results stated in \cref{sec:helly} we immediately obtain the algorithmic results promised in \cref{sec:intro}.
Note that the results hold not only for \rDomSet, but even for every domination-type problem of fixed radius $r$ and size $k$ that is considered the parameter.

\pagebreak
   \begin{theorem}\label{thm:ds-nd}
        Fix $r\in \N$. Then any
        domination-type problem 
        defined by a positive distance formula of size~$k$ and radius at most $r$ can be solved in time $2^{\Oh(k^2\log k)}\cdot \|G\|$ on any graph class $\Cc$
        such that either $\Cc$ is nowhere dense, or $\Cc=\Dd^s$ for a nowhere dense class $\Dd$ and some $s\in \N$, or $\Cc$ is the class of map graphs, or $r=1$ and $\Cc$ is the class of $K_{t,t}$-free graphs
        for some fixed $t\in \N$.
        Moreover, if this domination-type problem is \rDomSet for parameter $k$, then the running time can be improved to $2^{\Oh(k\log k)}\cdot \|G\|$.
    \end{theorem}
\begin{proof}%[Proof of Theorem~\ref{thm:ds-nd}]
By \cref{thm:semi-ladder}, the class $\delta_r(\Cc)$ has finite semi-ladder index.
By \cref{lem:num-profiles} and its strengthenings, \cref{lem:prof-comp-powers} and \cref{lem:prof-comp-ktt}, $\numprofiles{r}{\Cc}{m}$ is bounded by a polynomial in~$m$.
Hence, we may apply \cref{thm:semi-ladder-algo}; the claimed running times follow from the remark following it.
\end{proof}
        
 %   \begin{theorem}\label{thm:ds-ktt}
 %           For any fixed $t\in \N$, any dominating-type problem of radius $1$ and size $k$ can be solved
%            on $K_{t,t}$-free graphs in time $2^{\Oh(k\log k)}\cdot \|G\|$.
%    \end{theorem}
%\begin{proof}%[Proof of Theorem~\ref{thm:ds-ktt}]
%Let $\Cc$ be the class of $K_{t,t}$-free graphs.
%By Theorem~\ref{thm:semi-ladder}, the class $\delta_1(\Cc)$ has finite semi-ladder index.
%By Lemma~\ref{lem:prof-comp-ktt} (see Appendix~\ref{sec:omitted}), $\numprofiles{1}{\Cc}{m}$ is polynomial in $m$.
%Again, we may hence apply Theorem~\ref{thm:semi-ladder-algo}.
%\end{proof}

We now move to the independence problems, for which we apply the Ladder algorithm.

\begin{theorem}\label{thm:ladder-thm}
    Let $r\in \N$ and let $\Cc$ be a class of graphs such that for any $k\in \N$, the class $\eta^k_r(\Cc)$ has ladder index smaller than $\ell(k)$ and has the weak $p(k)$-Helly property, 
    for some functions $\ell,p\colon \N\to \N$. Then the \rIndSet problem on $\Cc$ can be solved in time $f(k)\cdot \|G\|$, for some function $f$.
    \end{theorem}
\begin{proof}%[Proof of Theorem~\ref{thm:ladder-thm}]
Given a graph $G$, we apply the Ladder Algorithm in the graph $\eta^k_r(\Cc)$ with implementations of oracles provided by \cref{lem:oracles}.
The correctness of the algorithm and the running time bound follow directly from \cref{cor:ladder-algo} and \cref{lem:oracles}, where we may set 
$f(k)=\Oh(R^{p(k)}(2\ell(k))\cdot \numprofiles{r}{\Cc}{p(k)\cdot R^{p(k)}(2\ell(k)))}^{p(k)\cdot k})$.
\end{proof}

\begin{theorem}\label{thm:is-nd}
For any $r\in \N$ and nowhere dense class $\Cc$, the \rIndSet problem on $\Cc$ can be solved in time $f(k)\cdot \|G\|$, for some function $f$.
\end{theorem}
\begin{proof}%[Proof of Theorem~\ref{thm:is-nd}]
By \cref{thm:stability} and~\ref{thm:ind-weak}, for every $k\in \N$ there are constants $\ell,p\in \N$, 
depending on $k$, such that the class $\eta^k_r(\Cc)$ has ladder index bounded by $\ell$ and has the weak $p$-Helly property.
This allows us to apply \cref{thm:ladder-thm}.
\end{proof}

\paragraph*{Discussion of related results.}
% We now put the algorithmic results of Theorems~\ref{thm:ds-nd} and~\ref{thm:is-nd} in the context of the existing literature.
Fixed-parameter tractability of both \rDomSet and \rIndSet on any nowhere dense class follows from the general model-checking result for first-order logic of Kreutzer et al.~\cite{GroheKS17}.
The algorithms derived in this manner have running time $f(k)\cdot n^{1+\eps}$ for any fixed $\eps>0$ and some function $f$, where $n$ is the number of vertices of the input graph.
In fact, an algorithm with running time $f(k)\cdot n^{1+\eps}$ for the \rIndSet problem is one of the intermediate results used in~\cite{GroheKS17}.
A close inspection of this algorithm reveals that the polynomial factor is in fact $\|G\|$, improving the claimed $n^{1+\eps}$, however this is not explicit in~\cite{GroheKS17}.
For the \rDomSet problem, its fixed-parameter tractability on any nowhere dense class was established earlier by Dawar and Kreutzer~\cite{DawarK09},
but their algorithm had at least a quadratic polynomial factor in the running time bound.

As far as {\sc{Distance-$r$ Dominating Set}} on powers of nowhere dense classes is concerned, 
we remark that the result provided in \cref{thm:ds-nd} would {\em{not}} follow immediately from applying the algorithm on the graph before taking the power, for radius $rs$ instead of $r$.
The reason is that the input consists only of the graph $G^s$, and it is completely unclear how to algorithmically find the preimage $G$ if we are dealing with an arbitrary nowhere dense class~$\Dd$.
To the best of our knowledge, this result is a completely new contribution.

Regarding map graphs,
the fixed-parameter tractability of the {\sc{Distance-$r$ Dominating Set}} problem on this class of graphs was established by Demaine et al.~\cite{DemaineFHT05}.
However, they use the recognition algorithm for map graphs of Thorup~\cite{Thorup98} to draw a map model of the graph; this algorithm has an estimated running time of at least $\Oh(n^{120})$~\cite{Chen01a} and 
not all technical details have been published. Another way of obtaining a fixed-parameter algorithm would be to use the fact that map graphs have {\em{locally bounded rank width}}; 
%which follows from the
%observation of Chen et al.~\cite{ChenGP02} that map graphs are half-squares of planar graphs, and the fact that planar graphs have locally bounded treewidth. 
however, again achieving linear running time would be difficult due to the need of computing
branch decompositions with approximately optimum rankwidth, for which the best known algorithms have cubic running time. In contrast, as we have shown, 
the Semi-ladder Algorithm solves the problem in linear fixed-parameter time without the need of having a map model provided. 
%Recent characterizations of powers of bounded expansion classes~\cite{GajarskyKNMPST18,kwon2017low} fall short of giving efficient algorithms.

Finally, the fixed-parameter tractability of {\sc{Dominating Set}} on $K_{t,t}$-free graphs, where both $k$ and $t$ are considered parameters, was established by Telle and Villanger~\cite{TelleV12}.
Thus, \cref{thm:ds-nd} reproves this result and also improves upon the running time: from $2^{\Oh(k^{t+2})}\cdot \|G\|$ of~\cite{TelleV12} to $2^{\Oh(k\log k)}\cdot \|G\|$.

\section{Domination cores}\label{sec:cores}

As we remarked earlier, the argument used in the proof of \cref{lem:nd-dom-semiladder} is similar to the reasoning that shows that graphs from a fixed nowhere dense class admit small {\em{distance-$r$ domination cores}}: 
subsets of vertices whose distance-$r$ domination forces 
distance-$r$ domination of the whole graph; this concept was introduced by Dawar and Kreutzer~\cite{DawarK09}.
In fact, we can lift this result to the more general setting, for
domination-type properties with bounded semi-ladder indices, 
as follows.

Let $H=(L,R,E)$ be a bipartite graph.
As before, we consider the left side $L$ to be the set of 
candidates and the right side $R$ to be the set of witnesses.
An edge between a candidate $a\in L$ and a witness $b\in R$ is interpreted as that $a$ and $b$ agree, that is, $b$ agrees 
that $a$ is a solution. 
Here, the considered bipartite graph~$H$ will be of the form 
$\varphi(G)$ a some formula~$\varphi(\tup x;\tup y)$ expressing a domination-type problem.

A \emph{coverage core} for a set of witnesses $S\subset R$
is a subset $C\subset S$ such that 
every candidate $a\in L$ which agrees with all the witnesses from $C$,
actually agrees with all the witnesses from $S$. A coverage core for $H$ is a 
coverage core for $R$.

Our goal is to give an algorithm that for domination-type problems computes small coverage cores.
As before, we allow our algorithms a restricted access to $H$ 
via an oracle.

\bigskip

\begin{minipage}{0.96\textwidth}
\noindent{\bf{Semi-ladder Extension Oracle}}:
Given a set of witnesses $B\subseteq R$, the oracle either 
returns a candidate $a\in L$ and a witness 
$b\in R\setminus B$ such that
\begin{itemize}
\item $a$ and $b$ do not agree, that is, $(a,b)\not\in E$, and
\item $a$ agrees with all witnesses in $B$, 
\end{itemize}
or concludes that no such pair $a,b$ exists.
\end{minipage}

\bigskip

Before we show that Semi-ladder Extension Oracles for domination-type
properties can be efficiently implemented, we present an algorithm 
that computes a coverage core for $H$, 
given that $H$ has bounded semi-ladder index.

\medskip
\noindent\textbf{Coverage Core Algorithm.}
The Coverage Core Algorithm maintains 
two sequences $a_1,\ldots,a_n\in L$,
$b_1,\ldots,b_n\in R$ which form a semi-ladder in $H$.
Initially, both sequences are empty and $n=0$. 
 The algorithm repeats the following step.

 \bigskip

\begin{minipage}{0.96\textwidth} \noindent
\noindent{\bf{Extension step:}}
Let $B=\set{b_1,\ldots,b_n}$.
Apply the Semi-ladder Extension Oracle to the set $B$.
If the oracle does not return a pair,  
terminate and return~$B$ as a coverage core for $H$. 
Otherwise, let $(a,b)$ be the pair returned by the oracle.
Extend the semi-ladder by inserting $a$ as 
the element $a_{n+1}$ to the sequence  $a_1,\ldots,a_n$
and $b$ as the element $b_{n+1}$ to the sequence 
$b_1,\ldots,b_n$. Note that the sequences again 
form a semi-ladder in $H$. Proceed to the next round. 
\end{minipage}

\bigskip

\noindent The following lemma follows immediately from the fact 
that the Coverage Core Algorithm constructs in $n$ steps a semi-ladder
of length $n$. 

\begin{lemma}\label{lem:core-running-time}
The Coverage Core Algorithm applied to a bipartite 
graph $H$ with 
semi-ladder index
$\ell$ terminates after performing at most $\ell+1$ rounds. 
Consequently, it uses at most $\ell+1$ Semi-ladder Extension Oracle
calls, each involving a set of candidates $A$ with $|A|\leq \ell$ and a set of witnesses $B$ with $|B|\leq \ell$. 
\end{lemma}
%\begin{proof}
%Observe that in each round, one of the numbers
%$i$ or $n$ is incremented. Moreover,
%since $H$ has semi-ladder index $\ell$,
%the number $n$ remains bounded by $\ell$.
%Hence at some point, $i>n$ and the algorithm terminates,
%with the sizes of both sets $A$ and $B$ bounded by 
%$\ell$ throughout the algorithm. 
%\end{proof}

\begin{lemma}
The Coverage Core Algorithm applied to a bipartite graph $H$ 
returns a coverage core for $H$. 
\end{lemma}
\begin{proof}
The algorithm terminates if the Semi-ladder Extension 
Oracle does not return a candidate $a\in L$ and a 
witness $b\in R\setminus B$ for the set
$B=\set{b_1,\ldots,b_n}$. Unravelling the definition, 
this means that every candidate $a\in L$ which agrees with all witnesses from~$B$ also agrees with all witnesses from~$R\setminus B$. This means that 
$B$ is a coverage core for $H$. 
\end{proof}

Finally, we give an implementation of the Semi-ladder Extension
Oracle. 

\begin{lemma}\label{lem:le-extension-oracle}
Fix a distance formula $\phi(\bar x;\bar y)$ of radius $r$ with 
$|\bar x|=c$ and $|\bar y|=d$. Then for an input graph 
$G=(V,E)$, there is an implementation of Semi-ladder Extension 
Oracle calls in $\phi(G)$ with parameters
$B\subseteq V^{\bar y}$ running in time
\[\Oh\big(|G|^d\cdot (|B|\cdot \|G\|+|B|\cdot\numprofiles{r}{G}{d|B|+d}^{c})\big).\]
\end{lemma}
\begin{proof}
We iterate over all valuations $\bar b\in V^{\bar y}\setminus B$, 
giving a factor $|G|^d$ in the running time. Fix such 
$\bar b\in V^{\bar y}\setminus B$. Our task is to find 
$\bar a\in V^{\bar x}$ which
agrees with all witnesses from $B$, but which does not agree with $\bar b$. This can be done
by a slight modification of the Candidate Oracle (applied to $B\cup \{\bar b\}$), as presented in
\cref{lem:oracles}, running in time $\Oh(|B|\cdot \|G\|+|B|\cdot\numprofiles{r}{G}{d|B|+d}^{c})$. 
The modification boils down to ignoring $c$-tuples of profiles on the set of vertices involved in $B\cup \{\bar b\}$ that imply agreeing with $\tup b$. 
In total, we get the claimed  running time.
\end{proof}

\begin{corollary}
Fix $r\in \N$ and let $\Cc$ be a class of graphs such that for
each $q\leq r$, the class $\delta_q(\Cc)$ has
finite semi-ladder index. Then, for any positive distance formula
$\phi(\bar x, \bar y)$ of radius at most $r$, size at most $k$ 
and $|\bar y|=d$, 
the domination-type problem corresponding to $\phi$ on $\Cc$
admits a coverage core of size $f(k)$ and such a core can be computed 
in time $f(k)\cdot |G|^d\cdot \|G\|$, for some
function $f$. 
\end{corollary}
\begin{proof}
W.l.o.g. we can assume that $|\tup x|\leq k$.
Let $\ell\in \N$ be such that $\delta_q(\Cc)$ has semi-ladder index smaller than $\ell$, for all $q\le r$.
Given a graph $G$, we apply the Coverage Core Algorithm in the graph $H=\formulaGraph{\phi}{G}$ with implementation of the Semi-ladder Extension Oracle provided by \cref{lem:le-extension-oracle}.
By \cref{lem:bool-comb} we conclude the semi-ladder index of $\formulaGraph{\phi}{\Cc}$ is bounded by~$R^k(\ell)$.
Now the claimed running time follows immediately from \cref{lem:core-running-time} and \cref{lem:le-extension-oracle}.
\end{proof}

\section{Helly property for domination}\label{sec:domination}

In this section are going to present the proof of \cref{thm:semi-ladder}.
We start with the case when the considered class $\Cc$ is nowhere dense, and for this we use the well-known characterization of nowhere denseness in terms of {\em{uniform quasi-wideness}}.

\begin{definition}
  Let $s\colon \nats\to \nats$ and
  $N\colon \nats \times \nats \to \nats$ be functions. We say that a
  graph class~$\Cc$ is {\em{uniformly quasi-wide}} with
  {\em{margins}} $s$ and $N$ if for all $r,k\in \nats$, every graph
  $G\in \Cc$, and every vertex subset $W\subseteq V(G)$ of
  size larger than $N(r,k)$, there exist disjoint vertex subsets
  $S\subseteq V(G)$ and $A\subseteq W$ such that $|S|\leq s(r)$,
  $|A|>k$, and $A$ is distance-$r$ independent in $G-S$. A~class $\Cc$ is
  \emph{uniformly quasi-wide} if it is {uniformly quasi-wide} with
  some margins.
\end{definition}

\begin{theorem}[\cite{nevsetvril2010first,sparsity}]\label{thm:nd-uqw}
  A graph class $\Cc$ is nowhere dense if and only if it is uniformly quasi-wide.
\end{theorem}

We proceed to the nowhere dense case, which is encapsulated in the following lemma.

\begin{lemma}\label{lem:nd-dom-semiladder}
For every $r\in \N$ and nowhere dense class $\Cc$, the class $\delta_r(\Cc)$ has a finite semi-ladder index.
\end{lemma}
\begin{proof}
  Fix $r\in \nats$ and a graph $G\in \Cc$.
  Suppose vertices $a_1,\ldots,a_\ell$ and $b_1,\ldots,b_\ell$ form a semi-ladder of order $n$ in $\delta_r(G)$, that is, we have that 
  \begin{itemize}
  \item $\dist_G(a_i,b_i)>r$ for each $i\in \{1,\ldots,\ell\}$; and
  \item $\dist_G(a_i,b_j)\leq r$ for all $i,j\in \{1,\ldots,\ell\}$ with $i>j$.
  \end{itemize}
  We need to give a universal upper bound on $\ell$, expressed only in terms of $r$ and $\Cc$. 
  
  By \cref{thm:nd-uqw}, $\Cc$ is uniformly
  quasi-wide, say with margins $s\colon \nats\to \nats$ and
  $N\colon \nats \times \nats \to \nats$.
  Let $s=s(2r)$ and $t=2\cdot (r+2)^{s}$.  
  Suppose from now on that $\ell>N(2r,t)$. 
  
  Let $W=\set{a_1,\ldots,a_\ell}$; then $|W|=\ell>N(2r,t)$, because vertices $a_i$ are pairwise different. By
  uniform quasi-wideness, we can find disjoint vertex subsets
  $S\subseteq V(G)$ and $A\subseteq W$ such that $|S|\leq s$, $|A|>t$,
  and~$A$ is distance-$2r$ independent in $G-S$.

  As $|S|\leq s$, there are at most $(r+2)^s$ possible distance-$r$ profiles on $S$ (called further just {\em{profiles}} for brevity).
  Since $|A|>t=2\cdot (r+2)^s$, we can find 
  three indices $1\le \alpha<\beta<\gamma\leq \ell$ such that the vertices
  $x:=a_\alpha,y:=a_\beta,z:=a_\gamma$ belong to $A$ and have equal profiles.
  Denote $w:=b_\alpha$. In particular, we have the following:
  \begin{itemize}
  \item the distance between $y$ and $z$ in $G-S$ is larger than $2r$;
  \item the vertices $x,y,z$ have the same profiles; and
  \item $\dist_G(w,x)>r$,  $\dist_G(w,y)\le r$, and $\dist_G(w,z)\le r$.  
\end{itemize}
\begin{figure}[ht]
	\centering
                \centering
                \def\svgwidth{0.4\columnwidth}
                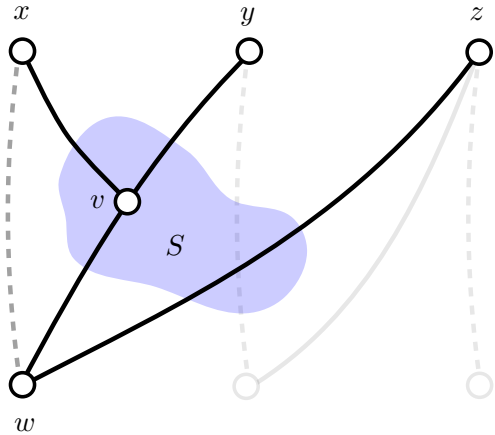
	\caption{Proof of \cref{lem:nd-dom-semiladder}: contradiction looms.}
	\label{fig:uqw}
\end{figure}

We show that this leads to a contradiction (see \cref{fig:uqw}).
Let $P_{yw}$ be a path of length at most $r$ connecting~$y$ and $w$, and
let $P_{wz}$ be a path of length at most $r$ connecting $w$ and~$z$. In
particular, the concatenation of $P_{yw}$ and $P_{wz}$ has length at most $2r$ and connects
$y$ and~$z$. Since the distance between $y$ and $z$ in $G-S$ is larger than $2r$,
at least one of the paths $P_{yw},P_{wz}$ must contain a vertex $v\in S$.
Suppose that it is $P_{yw}$, the other case being analogous. Then $P_{yw}$
is split by $v$ into two subpaths: $P_{yv}$ and $P_{vw}$.  Since~$x$ and $y$
have the same profiles, we may find a path $P_{xv}$ connecting $x$ and $v$
whose length is not larger than the length of $P_{yv}$. Now, the concatenation of paths $P_{xv}$ and $P_{vw}$ has
length at most $r$ and connects $x$ with $w$, a~contradiction with $\dist_G(w,x)>r$.

Therefore, we must have that $\ell\leq N(2r,2\cdot (r+2)^{s(2r)})$, which finishes the proof.
\end{proof}

The case when $\Cc$ is a power of a nowhere dense class now follows almost immediately.

\begin{lemma}\label{lem:pow-dom-semiladder}
For all $r,s\in \N$ and nowhere dense class $\Dd$, the class $\delta_r(\Dd^s)$ has a finite semi-ladder index.
\end{lemma}
\begin{proof}
It suffices to observe that $\delta_r(\Dd^s)=\delta_{rs}(\Dd)$ and use \cref{lem:nd-dom-semiladder}.
\end{proof}

For the case of map graphs, we use the following characterization of map graphs as {\em{half-squares of planar graphs}}, due to Chen et al.~\cite{ChenGP02}.

\begin{theorem}[\cite{ChenGP02}]\label{thm:half-squares}
If $G=(V,E)$ is a map graph, then there exists a bipartite planar graph $H$ with $V$ being one side of its bipartition such that $G$ is the subgraph induced by $V$ in $H^2$.
\end{theorem}

\begin{lemma}\label{lem:map-dom-semiladder}
Let $\Cc$ be the class of map graphs. Then for every $r\in \N$, the class $\delta_r(\Cc)$ has a finite semi-ladder index.
\end{lemma}
\begin{proof}
Let $\Pc$ be the class of planar graphs. Since $\Pc$ is nowhere dense, by \cref{lem:pow-dom-semiladder} we have that the class $\delta_r(\Pc^2)$ has finite semi-ladder index.
It now suffices to observe that by \cref{thm:half-squares}, every graph in $\delta_r(\Cc)$ is an induced subgraph of a graph in $\delta_r(\Pc^2)$, and semi-ladder index is monotone under taking induced subgraphs.
\end{proof}

We are left with the case of $r=1$ and classes excluding a fixed complete bipartite graph.

\begin{lemma}\label{lem:ktt-dom-semiladder}
  Let $t\in \N$ and let $\Cc$ be the class of $K_{t,t}$-free graphs.
  Then the semi-ladder index of the class~$\delta_1(\Cc)$ is smaller than $3t$.
\end{lemma}
\begin{proof}
  Suppose  for the sake of contradiction that there exists a graph $G\in \Cc$ and vertices $a_1,\ldots,a_{3t}$ and $b_1,\ldots,b_{3t}$ that form a semi-ladder of order $3t$ in $\delta_1(G)$.
  Note here that $\delta_1(x,y)$ checks whether $x$ and $y$ are equal or adjacent, which means that $a_i$ and $b_j$ are equal or adjacent whenever $i>j$, 
  whereas $a_i$ and $b_i$ are different and non-adjacent for all $i$.
  Observe that $b_i\neq b_j$ for all
  $1\le j<i\leq 3t$, because~$\delta_1(a_i,b_j)$ holds, while $\delta_1(a_i,b_i)$ does not.
  Similarly, $a_i\neq a_j$ for all $1\le j<i\leq 3t$.
  Hence, among vertices $a_{t+1},\ldots, a_{3t}$ there are at 
  least $t$ vertices different from the vertices $b_1,\ldots, b_s$.
  These vertices on one side, and vertices $b_1,\ldots, b_t$ on the other side, form a $K_{t,t}$ subgraph in $G$,
  contradicting the assumption that $G\in \Cc$.
\end{proof}

\cref{lem:nd-dom-semiladder},~\ref{lem:pow-dom-semiladder},~\ref{lem:map-dom-semiladder}, and~\ref{lem:ktt-dom-semiladder} together prove \cref{thm:semi-ladder}.
\section{Helly property for independence}\label{sec:independence}

In this section we show that distance-$r$ independence has the
weak Helly property on every nowhere dense class of graphs.
More precisely, we prove \cref{thm:ind-weak} and derive some auxiliary algorithmic results from it.

Let us start with clarifying the notation. Let $G$ be a graph and let $r\in \N$.
For a graph $G$, by $V(G)$ and $E(G)$ we denote the vertex and the edge set of $G$, respectively.
For vertex $u\in V(G)$, the {\em{distance-$r$ neighborhood}} of $u$ is the set $N^G_r(u)=\{v\colon \dist_G(u,v)\leq r\}$.
This notation is extended to subsets by setting $N^G_r(U)=\bigcup_{u\in U} N^G_r(u)$.
A vertex subset $A\subseteq V(G)$ is {\em{distance-$r$ independent}} if $\dist_G(u,v)>r$ for all distinct $u,v\in A$.
A vertex subset $D\subseteq V(G)$ {\em{distance-$r$ dominates}} $A$ if $A\subseteq N^G_r(D)$.
Then $D$ is a {\em{distance-$r$ dominating set}} in $G$ if $D$ distance-$r$ dominates $V(G)$.

Throughout this section we will use distance-$r$ profiles, which we discussed in \cref{sec:algo}.
We extend the notation for profiles to subsets of vertices as follows:
For $r\in \N$, a graph $G$, and vertex subsets $S,U\subseteq V(G)$ we define 
$$\profile_r^{G,S}(U)(v)=\min_{u\in U}\ \profile_r^{G,S}(u)(v)\qquad\textrm{for all }v\in S.$$
Note that, again, $\profile_r^{G,S}(U)$ is a function from $S$ to $\{0,1,2,\ldots,r,\infty\}$, and there are $(r+2)^{|S|}$ possibilities for such a function.

%\subsection{Proof of Theorem~\ref{thm:ind-weak}}

\paragraph*{Dependence cores and pre-cores.}
We will work with the following combinatorial object that 
we call a \emph{dependence core}. 

\begin{definition}
Let $G$ be a graph and let $r,k\in \N$. A \emph{distance-$r$ 
dependence core} for $G$ and~$k$ is a subset of vertices $Q\subseteq V(G)$
with the following property: for every set $X\subseteq V(G)$ of size
at most $k$ there exists a path of length at most $r$ containing a 
vertex of $Q$ and connecting two elements of $X$. 
\end{definition}

Observe that %in contrast to a domination core which exists in every 
%graph (the whole vertex set is always a domination core), 
a distance-$r$ 
dependence core for $G$ and $k$ can only exist if $G$ does not
contain a distance-$r$ independent set of size $k$. 
By definition of $\eta_r^k$, \cref{thm:ind-weak} is an 
immediate consequence of the following theorem. 

\begin{theorem}\label{thm:dependence-core}
Let $r,k\in \N$ and $\Cc$ be a nowhere dense graph class.
Then there exists $p\in \N$, depending on~$r$,~$k$, and $\Cc$, such 
if $G\in\Cc$ does not contain a distance-$r$ independent set of size
$k$, then $G$ has a distance-$r$ dependence core of size at most $p$. 
\end{theorem}

Hence, it suffices to prove \cref{thm:dependence-core}. 
Towards this goal, we introduce the following definitions. 

\begin{definition}
Let $G$ be a graph, let $Q,D\subseteq V(G)$, and let $a\in V(G)$.
We shall say that $Q$ \emph{distance-$r$ captures} the 
pair $(D,a)$ if there exists a path of length at most $r$ containing
a vertex of $Q$ and connecting $a$ with a vertex $d$ of $D$.
\end{definition}

In the following we will always talk about distance-$r$ capturing for the number $r$ clear in the context, 
and hence for brevity we will just use the term {\em{capture}}.

\begin{definition}
Let $G$ be a graph, $A\subseteq V(G)$ 
and let $r,k\in \N$. A \emph{distance-$r$ 
dependence pre-core} for $G,A$ and~$k$ is a subset of 
vertices $Q\subseteq V(G)$ with the following property:
for every set $D\subseteq V(G)$ that has size at most $k$ and distance-$r$ dominates $A$,
the set $Q$ captures $(D,a)$ for each $a\in A$.
\end{definition}

While dependence cores can only exist in the absence of independent
sets of size $k$, pre-cores always exist (the whole vertex set is always a 
pre-core). 
In order to prove \cref{thm:dependence-core}, we prove
that in case that a graph $G$ does not contain a distance-$r$
independent set of size~$k$, a distance-$r$ dependence
pre-core for $G, V(G)$ and $k-1$ is in fact a distance-$r$ 
dependence core for $k$. We then proceed
to prove that nowhere dense graph classes admit small pre-cores. 
In fact, the first author proves in his thesis that 
monotone nowhere dense graph classes are characterized by 
the existence of small pre-cores for all vertex subsets $A$ and values of $r$. 

\paragraph*{A pre-core is a core, whenever possible.}
We first prove that every distance-$r$ dependence pre-core for
$G, V(G)$ and $k-1$ in absence of a distance-$r$ independent set 
of size $k$ is in fact
a distance-$r$ dependence core for $k$. 

\begin{lemma}\label{lem:pre-core}
Let $G$ be a graph and let $r,k\in \N$. Let $Q$ be a distance-$r$
dependence pre-core for $G, V(G)$ and~$k-1$. If $G$ does not contain a
distance-$r$ independent set of size $k$, then $Q$ is a distance-$r$
dependence core for $G$ and $k$. 
\end{lemma}
\begin{proof}
Assume towards a contradiction that $Q$ is not a distance-$r$ 
dependence core for $G$ and $k$. 
Hence, there exists $X\subseteq V(G)$ of size
$k$ such that $Q$ does not intersect any path of length at most $r$ between two distinct vertices of $X$.
From all such choices for~$X$, fix one that minimizes the function
\[ f(X) = \left| \{ w \in X \colon \text{ there exists $v\in X$, $v\neq w$,
with $\dist(v,w) \leq r$}\} \right|. \]
By assumption, there does not exist a distance-$r$ independent
set of size $k$ in $G$, so there are $v,w\in X$, $v\neq w$,
with $\dist(v,w)\leq r$. 

If $X-\{w\}$ is a distance-$r$ dominating set of $G$, then, by definition of a pre-core, 
for every vertex $u\in V(G)$ there exists a path of length at most $r$ containing a vertex of 
$Q$ and connecting~$u$ with a vertex of $X-\{w\}$. In particular, 
this holds for the vertex $w$. But then $Q$ crosses some path of length at most $r$ between two distinct vertices of $X$, a contradiction.

Otherwise, $X-\{w\}$ is not a distance-$r$ dominating set of $G$. 
Let $u$ be a vertex with $\dist(u,v)>r$ for all $v\in X-\{w\}$. 
We define $X'\coloneqq (X\cup\{u\})\setminus\{w\}$ and claim that
$X'$ satisfies the condition on $X$: no path of length at most $r$ connecting two distinct elements of $X'$ is intersected by $Q$.
To see this, assume that there exists such a path. As $u$ is not
within distance $r$ to any other vertex of $X'$, we have 
$x_1,x_2\neq u$. Hence, the considered path connects two 
vertices from $X$, contradiction our assumption on $X$. 
Now observe that $f(X')<f(X)$, as by construction of $X'$ we have
$f(X') = f(X - \{w\}) = f(X)-1$, contradicting the minimality of $X$. 
This finishes the proof.
\end{proof}

Now, \cref{thm:dependence-core} follows immediately
from \cref{lem:pre-core} and the following lemma, which we are going to prove next.

\begin{lemma}\label{lem:pre-core-exists}
Let $r,k\in \N$ and $\Cc$ be a fixed nowhere dense graph class. Then there
exists $p\in\N$, depending on $r,k$ and $\Cc$ such that for every 
$G\in\Cc$ and every $A\subseteq V(G)$ 
there exists a distance-$r$ dependence pre-core for $G$, $A$ and
$k$ of size at most $p$.
Moreover, given $G$, $A$, and $k$, such a distance-$r$ dependence pre-core can be computed in time $f(r,k)\cdot \|G\|$,
for some function $f(r,k)$ that is polynomial in $k$ for fixed $r$.
\end{lemma}

\paragraph*{Splitter game.}
In order to prove \cref{lem:pre-core-exists} we appeal 
to a characterization of nowhere dense graph classes in terms of a 
game, called the splitter game, which was introduced in~\cite{GroheKS17}.
 
\begin{definition}[Splitter game] Let $G$ be a graph and let 
$\ell, r \in \N$. The $\ell$-round radius-$r$ \emph{splitter game} 
on~$G$ is played by two players, \emph{Connector} and \emph{Splitter}, 
as follows. We let $G_0 := G$. In round $i + 1$ of the
game, Connector chooses a vertex $v_{i+1} \in G_i$. 
Then Splitter picks a vertex $w_{i+1} \in N_r^{G_i}(v_{i+1})$. We define 
$G_{i+1} \coloneqq G_i[N_r^{G_i}(v_{i+1}) - \{w_{i+1}\}]$. Splitter 
wins if $G_{i+1} = \emptyset$. Otherwise the game continues
on the graph~$G_{i+1}$. If Splitter has not won after~$\ell$ rounds, then Connector wins.
 \end{definition}
 
A strategy for Splitter is a function $f$ that associates to every 
partial play $(v_1, w_1,\ldots, v_i, w_i)$ with
associated sequence $G_0,\ldots, G_i$ of graphs and move 
$v_{i+1}\in V(G_i)$ by Connector a vertex $w_{i+1}\in N^{G_i}_r(v_{i+1})$.
A strategy $f$ is a winning strategy for Splitter in the $\ell$-round 
radius-$r$ splitter game on $G$ if Splitter wins every play in which he 
follows the strategy $f$. We say that a winning strategy is computable
in time $T$ if for every partial play as above, the vertex $w_{i+1}$ can
be computed in time $T$. 
 
\begin{theorem}[\cite{GroheKS17}]\label{thm:SplitterGame}
Let $\Cc$ be a class of graphs. Then $\Cc$ is nowhere dense if and 
only if for every $r \in \N$ there exists $\ell \in \N$, such that for 
every $G \in \Cc$, Splitter wins the $\ell$-round radius-$r$ splitter 
game on $G$. Furthermore, for fixed $r$, a winning strategy for splitter in the
radius-$r$ splitter game on any graph $G\in\Cc$ is computable in 
time $\mathcal{O}\big(\ell\cdot \|G\|\big)$. 
\end{theorem}

Now \cref{lem:pre-core-exists} is implied by \cref{thm:SplitterGame} and the following lemma (applied with $s=0$). 

\begin{lemma}\label{lem:precore-by-splitter}
For all $\ell, k, r, s \in \N$ there exists $p \in \N$, depending
polynomially on $k$ for fixed $\ell,r,s$, with the following property.
For every graph $G$ and $S \subseteq V(G)$ with $|S| \leq s$ and such that Splitter wins
the $\ell$-round radius-$3r$ splitter game on $G-S$, and
for every $A \subseteq G$, there exists a distance-$r$ dependence
pre-core $Q$ for $G,A$ and $k$ of size at most~$p$. Furthermore, 
such a set $Q$ is computable in time $f(\ell,k,r,s)\cdot \|G\|$,
for some computable function $f$ which is polynomial in $k$ for fixed $\ell,r,s$.
\end{lemma}

\paragraph*{Induction on the splitter's strategy.} 
Before we prove \cref{lem:precore-by-splitter}, we need one more lemma, 
which expresses the outcome of applying a naive greedy strategy for finding a small distance-$r$ dominating set.

\begin{lemma}\label{lem:greedy-is}
Let $G$ be a graph, $X\subseteq V(G)$ and $r, k \in \N$. Then either
\begin{enumerate}
 \item there is $Y \subseteq X$ such that $|Y| > k$ and $Y$ is distance-$2r$ independent in $G$, or
 \item there is $Z \subseteq X$ such that $|Z| \leq k$ and $X \subseteq N_{2r}^{G}(Z)$.
\end{enumerate}
 Moreover, there is an algorithm which for given $G, X, r, k$ computes $Y$ or $Z$ as above in time $\mathcal{O}(k\cdot \|G\|)$. 
\end{lemma}
\begin{proof}
We use a simple greedy algorithm, which aims to construct the set 
$Y$. Let us assume we have already constructed $Y \subset X$ such that $|Y| \leq k$ and $Y$ is distance-$2r$ independent. There are two possibilities:
\begin{itemize}
\item If $X \subset N_{2r}(Y)$, then we set $Z\coloneqq Y$ satisfying the second condition and terminate.
\item Otherwise, there is $v \in X$ such that $v \not\in N_{2r}(Y)$, so $Y \cup \{v\} $ is distance-$2r$ independent and larger, hence we continue with $Y \cup \{v\}$.
\end{itemize}
  After $k+1$ repetition of this procedure, we constructed a set $Y$ satisfying the first condition. Testing whether there exists $v \not\in N_{2r}(Y)$ can be easily done in linear time, so the final running time is $\mathcal{O}(k\cdot (n+m))$.
\end{proof}

\begin{proof}[of \cref{lem:precore-by-splitter}]
We prove the lemma by induction on the length $\ell$ of the splitter game. 
Furthermore, throughout the induction we maintain the following additional invariant: $Q\supseteq S$.
If $\ell=0$ we have $V(G)=S$ and we can set $p=s$ and $Q\coloneqq V(G)$. 

Denote by $\Pp$ the set of all distance-$r$ profiles on $S$; then $|\Pp|\leq (r+2)^s$.
In the induction step, our goal will be to find a family
$$\Gg = \{(G_t, S_t, A_t)\}_{t \in T}\qquad \textrm{for some index set }T,$$
where $G_t$ is a subgraph of $G$ and $S_t,A_t\subseteq V(G_t)$ satisfying the following properties:
  \begin{enumerate}
  \item The size $|\Gg|$ can be bounded by a function of $\ell,k,r,s$,
  which is polynomial in $k$ for fixed $\ell,r,s$.
  \item We have that $S_t\supseteq S$ and $|S_t|\leq |S|+1$ for all $t\in T$. 
  \item For each $t\in T$, Splitter wins the $\ell-1$-round
  radius-$3r$ splitter game on $G_t-S_t$. So by the induction hypothesis, 
  there exists a small distance-$r$ dependence pre-core $Q_t\supseteq S_t$ for $G_t$, $A_t$, and $k$. 
  \item For every $a \in A$ and for every set $D\subseteq V(G)$ that has size at most $k$ and distance-$r$ dominates $A$, 
  if the pair $(D,a)$ is not captured by $S$, then $(D, a)$ is captured by $Q_t$ in some graph~$G_t$. 
\end{enumerate}

As we will see, once we have constructed the family $\Gg$ we can output 
$Q\coloneqq \bigcup_{t \in T} Q_t$. Note that then indeed $Q\supseteq S$, as $Q_t\supseteq S_t\supseteq S$ for all $t\in T$.

\medskip
Each of the graphs $G_t$ is constructed by simulating one step
of the splitter game. That is, we consider some vertex $v_t\in G-S$
as the choice of Connector in the splitter game. By assumption, 
we can find an answer $w_t\in N_{3r}^{G-S}(v_t)$ of Splitter, so 
that Splitter wins the $\ell-1$-round radius-$3r$ splitter game on 
$G[N_{3r}^{G-S}(v_t) - \{w_t\}]$. We let $G_t\coloneqq 
G[N_{3r}^{G-S}(v_t) \cup S]$ and $S_t \coloneqq S \cup \{w_t\}$. 
Hence, the second and third of the above properties are satisfied for all 
constructed graphs $G_t$. Our task will be to identify a small set of 
moves $v_t$ of Connector and sets $A_t$, so that also the first and 
fourth property are satisfied. 

To construct such a set of moves of Connector we use a localization 
property for any pair $(D, a)$ to be captured --- we want to find a small 
set $Z$ such that every pair $(D, a)$ that 
needs to be captured will be close to some vertex of $Z$. This 
is made precise by the following statement, whose proof relies on \cref{lem:greedy-is}.

\begin{claim} \label{cl:localLemma}
There exists a set $Z\subseteq V(G)$ with $|Z|\leq k\cdot (r+2)^s$, 
such that for every set $D$ such that $|D|\leq k$ and $D$ distance-$r$ dominates of $A$, and every $a\in A$, 
if $(D,a)$ is not captured by $S$, then $a\in N_{2r}^{G-S}(Z)$.
\end{claim}
\begin{proof}
Fix $D$ and $a$ as in the statement of the claim. Let $p\coloneqq
\profile^{G,S}_r(a)$ be the distance-$r$ profile of $a$ on $S$. Let 
 \[T_p \coloneqq \{ x \in A : \profile_r^{G,S}(x) = p\}. \]
We apply \cref{lem:greedy-is} with parameters $G-S, T_p, r$ and $k$. 
The output is either a set $Y_p\subseteq T_a$ of size greater than $k$
which is distance-$2r$ independent in $G-S$, or a set $Z_p$ of size at most $k$ that distance-$2r$ dominates~$T_p$. 
In case the output is a set $Y_p$, we let $Z_p\coloneqq \emptyset$. 
We define 
 \[ Z \coloneqq \bigcup_{p \in \Pp} Z_p.\]
Then $|Z|\leq k\cdot |\Pp|\leq k\cdot (r+2)^s$, as required
It remains to prove that $a \in N^{G-S}_{2r}(Z)$.

Assume first that $Z_p\neq \emptyset$. Then $Z_p$ is distance-$2r$
dominates $T_p$ in $G-S$, in particular, the vertex $a$ (with
profile $p$) lies
in the distance-$2r$ neighborhood of $Z_p$. 

To finish the proof we show that the case $Z_p=\emptyset$ cannot
occur. Assume otherwise that there is $Y_p\subseteq T_p$ with
$|Y_p|>k$ and which is distance-$2r$ independent, i.e. the 
distance-$r$ neighborhoods $N_r^{G-S}(y)$, for all $y \in Y_p$, are pairwise disjoint.
By assumption, $|D| \leq k$, so for at least one $y_0 \in Y_p$, the set
$N_r^{G-S}(y_0)$ is disjoint with $X$. We claim that $\dist_G(y_0, D) > r$. 
We already know that this fact holds in $G-S$. The only remaining 
thing to show is the impossibility of connecting $y_0$ with~$D$ by a 
path of length at most $r$ through $S$. But this is impossible by the 
assumption that $S$ does not capture $(D, a)$, because $y$ and $a$ 
have the same distance-$r$ profiles on~$S$.
\cqed\end{proof}

We now use the set $Z$ provided by \cref{cl:localLemma} to 
generate the family $\Gg$. Indeed, if for each $z\in Z$ we let~$w_z$ be the Splitter's response for Connector's move $z$ in the radius-$3r$ game on $G$, then we may
define
\[G_z \coloneqq G[N^{G-S}_{3r}(z)\cup S]\qquad \textrm{and}\qquad S_z \coloneqq S \cup \{w_z\}.\]
With such definition, $G_z$ and $S_z$ satisfy the assertion of the induction hypothesis.
However, for a single $z\in Z$ we will consider multiple sets $A_t$, as explained next.

Define the index set for the family $\Gg$ as $T=Z\times \Pp$.
For $z\in Z$ and $p\in \Pp$, we set
\[A_{z,p} \coloneqq \{ a\in A \colon \profile_r^{G,S}(a)(v)+p(v)>r \textrm{ for all }v\in S\}\cap N_{2r}^{G-S}(z).\]
In other words, we define $A_{z,p}$ by taking $A\cap N_{2r}^{G-S}(z)$, and removing all vertices that are connected by a path of length at most $r$ passing through $S$ to any vertex with profile $p$ on $S$.
Note that in this definition we use distance $2r$ instead of $3r$.
Now, for each $t=(z,p)\in T$, we put
\[(G_t,S_t,A_t)=(G_z,S_z,A_{z,p}).\]
This defines the family $\Gg=\{(G_t,S_t,A_t)\}_{t\in T}$.

\begin{comment}
Fix any enumeration of $S$ as $(s_1,\ldots, s_s)$. 
Let $X\subseteq V(G_t)$. We let 
$\mathrm{profile}(G_t,X,S)\coloneqq (d_1,\ldots, d_s)$, 
where $d_i=\min_{x\in X}(\dist_{G_t}(s_i, x), r+1)$. 
Observe that there exit at most $(r+2)^s$ many different
profiles, independent of the size of $X$. We write $\Pp_t(S)$
for the set of all different profiles in $G_t$. Furthermore, 
observe that the profile of a set $X$ can be computed
in time $\mathcal{O}(s\cdot (n+m)$ by performing a breadth
first search from every element of $s$, stopping whenever
the first element of $X$ is encountered, or when after $r$ steps
no element of $X$ was encountered. 

For each pair $G_t, S_t$, we compute a family 
$\{B_t^P\}_{P\in \Pp_t}$. 
Fix an element $(d_1,\ldots,d_s)\in \Pp$.
Define $B_t^P$ as follows. Start with $B=A\cap N_{2r}^{G-S}(v_t)$. 
Now, for each $s_i\in S$ such that $d_i=r'\leq r$, remove 
all vertices $v$ of $B$ with 
$\dist_{G-S}(s,v)\leq r-r'$. As observed above, the set of all 
profiles can be computed in linear time. Hence, 
we can also compute each $B_t^P$ in linear time. Furthermore, 
there are at most $(r+2)^{s}$ different profiles, hence
we can compute all of the sets $B_t^P$ efficiently.
By renaming the indices and sets $B_t^P$ we obtain the desired
index set $T$ and triples $(G_t,S_t,A_t)$. 

Now, let~$Q_t$ be a 
distance-$r$ dependence pre-core for $G_t, A_t$ and $k$ of appropriate
size, which exists by induction hypothesis.
\end{comment}

Now, for each $t\in T$ apply the induction hypothesis to the triple $G_t,S_t,A_t$, yielding
a suitably small pre-core $Q_t$ for $G_t$, $A_t$ and $k$. We define 
$Q=\bigcup_{t\in T}Q_t$ and verify that it has all the required properties.

\begin{claim}
The set $Q$ is a distance-$r$ dependence pre-core for $G$, $A$ and $k$. 
\end{claim}
\begin{proof}
Let $D$ be a set of size at most $k$ that distance-$r$ dominates $A$, and let $a\in A$. 
Since $S$ is contained in $Q$, it suffices to prove the following: if $(D,a)$ is not
captured by $S$, then there is $t\in T$ such that $(D,a)$ is 
captured by~$Q_t$. By the construction of $Z$ there exists $z\in Z$ with $\dist_{G-S}(z,a)\leq 2r$. 
Let $p=\profile^{G,S}_r(D)$. 
Observe that since $(D,a)$ is not captured by $S$, we have $a\in A_{z,p}$.
We claim that $(D,a)$ is captured by $Q_t$ for $t=(z,p)$.

Let $D_t\coloneqq D\cap V(G_t)$. We verify that $D_t$ distance-$r$ dominates $A_t$ in $G_t$.
For this, take any $b\in A_t$. Since $D$ distance-$r$ dominates $A$ in $G$, 
there is some $d\in D$ such that $\dist_G(b,d)\leq r$.
As $b\in A_t$, in fact we must have $\dist_{G-S}(b,d)\leq r$.
This, together with the assertion $\dist_{G-S}(z,b)\leq 2r$ following from $b\in A_t\subseteq N^{G-S}_{2r}(z)$,
entails that $d\in N_{3r}^{G-S}(z)$, which in turns implies that
$d\in V(G_t)$ and $d$ distance-$r$ dominates $b$ in $G_t$. Therefore $d\in D_t$ and, consequently, $D_t$ distance-$r$ dominates every $b\in A_t$ in $G_t$.

Applying the induction assumption we conclude that $Q_t$ captures $(D_t,a)$ in $G_t$.
This implies that $Q$ captures $(D,a)$ and we are done.
\cqed\end{proof}

It remains to bound the size of $Q$, which we do as follows as follows. Let $c(r,k,d,s)$ be 
the smallest number for which the statement of the lemma holds. 
From the proof we obtain the following recursive bound on the function $c$:
\begin{align*}c(r,k,\ell,s) & \leq |Z|\cdot (r+2)^{s}\cdot c(r,k,\ell-1,s+1) \leq k\cdot (r+2)^{2s}\cdot c(r,k,\ell-1,s+1); \\
 c(r,k,0,s) & \leq s. 
\end{align*}
This gives an upper bound $c(r,k,\ell,s) \leq k^\ell\cdot (r+2)^{2\ell(s+\ell)}\cdot (s+\ell)$, which is indeed polynomial in $k$ for fixed $r,\ell,s$.

Furthermore, the provided proof is not only constructive, but effectively computable by a recursive algorithm.
In the following, by when speaking about linear time we mean time of the form $f(k,\ell,r,s)\cdot \|G\|$ for a function $f$ that is polynomial in $k$ for fixed $\ell,r,s$.
\begin{itemize}
\item Computing $Z$ using \cref{lem:greedy-is} can be done in linear time.
\item Computing Splitter's response $w_z$ to each possible Connector's move $z\in Z$ can be done in linear time. 
Computing the resulting graphs $G_z$ and sets $S_z$ can be done in linear time. 
\item For each resulting pair $G_z,S_z$, the sets 
$A_{z,p}$ for $p\in \Pp$ can be computed in linear time, by applying breadth-first search from each vertex of $S$.
\item Finally, for each resulting triple $(G_t,S_t,A_t)$ we make a recursive subcall.
\end{itemize}
Hence, we have a recursion of depth at most $\ell$ with a branching
of order $|T|$, which is bounded by a function that is polynomial in $k$. Hence, 
in total, we have a running time $f(k,\ell,r,s)\cdot \|G\|$ for a 
function $f$ which is polynomial in $k$.
\end{proof}

As argued, \cref{lem:precore-by-splitter} implies \cref{lem:pre-core-exists}, which implies \cref{thm:dependence-core}, which implies \cref{thm:ind-weak}, and we are done.

\paragraph*{A linear time algorithm for \rIndSet
in nowhere dense classes.}
We now we apply the obtained results to construct an algorithm 
for \rIndSet on nowhere dense graph classes via a direct dependence
pre-core search. Precisely, we prove the following theorem. 

\begin{theorem}\label{thm:ind-fpt}
Let $\Cc$ be a nowhere dense class of graphs an let $r\in \N$. 
Then there exists an algorithm which on input $G\in \Cc$ and $k\in\N$ 
decides in time $2^{\Oh(k\log k)}\cdot \|G\|$ whether $G$ contains a distance-$r$ independent
and outputs such a set, or correctly decides that no such set exists. 
\end{theorem}
\begin{proof}
By \cref{lem:pre-core-exists}, in time $\poly(k)\cdot \|G\|$ we can compute a distance-$r$ pre-core $Q$ for $G$, $V(G)$, and $k-1$ of size $\poly(k)$.
By \cref{lem:pre-core}, whether $G$ contains a distance-$r$ independent set is equivalent to whether $Q$ is a distance-$r$ core for $G$ and $k$.

We can check the latter assertion as follows. 
Observe that $Q$ is a core as above if the following condition holds for every set $X\subseteq V(G)$ with $|X|=k$: 
there exist $x_1,x_2\in X$, $x_1\neq x_2$, and $y\in Q$ such that $\dist(x_1,y)+\dist(x_2,y)\leq r$.
This condition depends only on the multiset $\{\{ \profile_r^{G,Q}(x)\colon x\in X\}\}$, 
and given such a multiset of size $k$ it can be decided in time $\Oh(|Q|\cdot k^2)$ whether the corresponding set $X$ indeed satisfies the condition.
Observe that we can construct the multiset $\Pp=\{\{ \profile_r^{G,Q}(u)\colon u\in V(G)\}\}$ in time 
$\Oh(|Q|\cdot \|G\|)$ by running BFS from every vertex of $Q$, reading for every vertex $u\in V(G)$ its distance-$r$ profile on $Q$, and then counting how many times each profile is realized.
Now, it remains to check all submultisets of $\Pp$ of size $k$. 
By \cref{lem:num-profiles}, the number of different distance-$r$ profiles on~$Q$ that are actually realized in $G$ is bounded polynomially in $|Q|$, so the number of such multisets is at most~$|Q|^{\Oh(k)}$. 
Since $|Q|=\poly(k)$, we have $2^{\Oh(k\log k)}$ multisets to check, and checking each of them takes time $\Oh(|Q|\cdot k^2)=\poly(k)$.
Hence, the total running time is $2^{\Oh(k\log k)}\cdot \|G\|$.

Observe that the above reasoning only gives a decision procedure and does not construct the actual distance-$r$ independent set.
Such a set can be constructed as follows.
Using the algorithm described above we can find a set $X$ with $|X|=k$ such that $\dist(x_1,y)+\dist(x_2,y)>r$ for all distinct $x_1,x_2\in X$ and $y\in Q$, or conclude that no such set $X$ exists.
In the latter case, by the reasoning above we conclude that there is no distance-$r$ independent of size~$k$ in $G$.
Otherwise, we emulate the reasoning presented in the proof of \cref{lem:pre-core}: supposing $X$ is not a distance-$r$ independent set yet, we can find another set $X'$ 
satisfying the same condition and with $f(X')<f(X)$, where $f$ is defined as in the proof of \cref{lem:pre-core}, and apply the reasoning again replacing $X$ with $X'$.
The number of iterations until the procedure terminates is bounded by $k$ and each iteration can be easily implemented in time $\Oh(k^2\|G\|)$.
Hence, a distance-$r$ independent set can be constructed using $\Oh(k^3\|G\|)$ additional time.
\end{proof}

\section{Omitted proofs}\label{sec:omitted}

Finally, in this section we present all proofs that we omitted from the main body of the text for the sake of a concise presentation. 

\begin{lemma}\label{lem:ramsey-precise}
For all $c,\ell\in \N$ with $c\geq 2$, we have $R^c(\ell)\leq c^{c\ell-1}$.
\end{lemma}
\begin{proof}
We first prove the following claim: if $G$ is a complete graph on a set $V$ of $c^{p-1}$ vertices, then
there is a sequence $u_1,\ldots,u_{p}$ of different vertices from $V$ with the following property:
for every $i\in \{1,\ldots,p\}$ all edges connecting $u_i$ with vertices appearing earlier in the sequence, i.e., $u_j$ for $j<i$, are of the same color.
Note that this color may vary for different vertices $u_i$.

We proceed by induction on $p$, where the base case $p=1$ holds by taking $u_i$ to be any vertex of $p$.
For the induction step, let $v$ be any vertex of $V$. Since edges between $v$ and $V\setminus \{v\}$ are colored with $c$ colors, there is a color $s\in \{1,\ldots,c\}$ such that
$v$ has at least $\left\lceil \frac{c^{p-1}-1}{c} \right\rceil=c^{p-2}$ neighbors adjacent via an edge of color $s$. Let $H$ be the complete graph induced by those neighbors in $G$.
Then by the induction hypothesis, in $H$ we can find a suitable sequence of vertices $u_1,\ldots,u_{p-1}$.
It now suffices to extend this sequence with $u_p=v$.

We proceed to the main proof. Using the claim, within any complete graph with $c^{c\ell-1}$ vertices and with edges colored with $c$ colors we may find a sequence $u_1,\ldots,u_{c\ell}$
such that for every $i\in \{1,\ldots,c\ell\}$, all edges connecting $u_i$ with vertices $u_j$ for $j<i$ are of the same color, say $s_i$.
Since there are $c$ colors in total, for some color $s$ there are at least $\ell$ indices $i\in \{1,\ldots,c\ell\}$ with $s_i=s$.
Then vertices $u_i$ corresponding to these indices $i$ form a monochromatic clique of color $s$.
\end{proof}

\begin{proof}[of \cref{lem:comatch&semiladder}]
The left-to-right implication is immediate:
every ladder of order $n$ and every co-matching of order $n$
are also semi-ladders of order $n$, so the semi-ladder index is always an upper bound on both the co-matching and the ladder index.

For the right-to-left implication, we prove that if a bipartite graph $G=(L,R,E)$ has both the ladder and the co-matching index smaller than some $\ell\in \N$, then its semi-ladder index is smaller than $q=R^2(\ell)$.
Suppose for contradiction that $a_1,\ldots,a_q\in L$ and $b_1,\ldots,b_q\in R$ form a semi-ladder of order $q$ in $G$.
Color all pairs $(i,j)$ with $1\le i<j\ge q$ red or blue, depending on whether $(a_i,b_j)\in E$ or not.
By Ramsey's theorem, there is a subset of $\set{1,\ldots,q}$
of size $\ell$ which is monochromatic, i.e., which only spans 
red edges, or only spans blue edges.
In the first case, the corresponding vertices $a_i$ and $b_i$ form a co-matching of order $\ell$ in $G$,
and in the latter case we analogously exhibit a ladder of order $\ell$ in $G$.
In any case, this is a contradiction.
\end{proof}

\begin{proof}[of \cref{lem:hel-com}]
    Let $G=(L,R,E)$.
    For the left-to-right implication,
    suppose $a_1,\ldots,a_q\in L$ and $b_1,\ldots,b_q\in R$ form a co-matching of order $q$ in $G$, for some $q\in \N$.
    Then $A=\{a_1,\ldots,a_q\}$ and $B=\{b_1,\ldots,b_q\}$ do not have the $q$-Helly property, implying that $q\leq p$.
    This means that the co-matching index of $G$ is at most $p$.

For the right-to-left implication, 
it is enough to show that if $G$ has co-matching index at most $p$, then it has the weak $p$-Helly property.
Indeed, from this it follows that in fact $G$ must have the strong $p$-Helly property, since we can apply the same argument to every induced subgraph of $G$, 
which also has co-matching index at most $p$.

Thus, assume that $G$ has co-matching index at most $p$ and $R$ is not covered (otherwise we are done).
Let $B\subset R$ any inclusion-minimal set which is not covered by $L$;
we show that $|B|\le p$.
By minimality, for every $b\in B$ there is some $a_b\in L$ 
which is adjacent to all vertices in $B\setminus \set b$ and not to $b$.
This means that~$B$ together with $A=\{a_b\colon b\in B\}$ form a co-matching of order $|B|$ in $G$,
which implies that $|B|\le p$ by the assumption on the co-matching index of $G$.
\end{proof}

\begin{proof}[of \cref{lem:bool-comb}]
    Let $G$ be a graph with vertex set $V$ and suppose that 
    $\psi(G)$ has semi-ladder index at least $R^k(\ell)$.    
    Then there are tuples $\tup a_1,\ldots,\tup a_q\in V^{\tup x}$ 
    and $\tup b_1,\ldots, \tup b_q\in V^{\tup y}$,
    where $q=R^k(\ell)$,
    such that for all $i,j\in \set{1,\ldots,q}$,
    if $i>j$ then $\psi(\tup a_i;\tup b_j)$ holds in $G$,
    and if $i=j$, then $\psi(\bar a_i;\bar b_i)$ does not hold in~$G$.  

    Since $\psi$ is a positive 
    boolean combination of $\phi_1,\ldots,\phi_k$,
    whenever  $(\tup a,\tup b),(\tup a',\tup b')\in V^{\tup x}\times V^{\tup y}$ 
    are two pairs of tuples, one satisfying $\psi$ in $G$ and 
    the other satisfying $\neg \psi$ in $G$,
    then there must be some $p\in \set{1,\ldots,k}$
    such that $(\tup a,\tup b)$ satisfies $\phi_p$ and
    $(\tup a',\tup b')$ satisfies $\neg \phi_p$.

Hence, for all $q\ge i>j\ge 1$,
we may color the pair $(i,j)$ with some number 
$p\in \set{1,\ldots,k}$ such that~$\phi_p(\tup a_i;\tup b_j)$ holds in $G$ 
and $\phi_p(\tup a_i;\tup b_i)$ does not hold in $G$.

By Ramsey's theorem, there 
is a set $X\subset \set{1,\ldots,q}$
of size $\ell$, and a color $p\in \set{1,\ldots,k}$,
such that each pair $(i,j)\in X^2$ with $i>j$ 
is colored with color $p$. It follows that
$\phi_p(a_i;b_j)$ holds in $G$ for all $i,j\in X$ with $i>j$, and $\varphi_p(\tup a_i;\tup b_i)$ does not hold in $G$ for all $i\in X$.
 This shows that $\phi_p(G)$ contains a semi-ladder of order $\ell$,
 contrary to the assumption.
\end{proof}

\begin{proof}[of \cref{lem:disj-comb}]
We prove the following statement: if for some graph $G=(V,E)$, in $\psi(G)$ we find a semi-ladder of order $q=k^{\ell-1}$, say
formed by $\tup a_1,\ldots,\tup a_q\in V^{\tup x}$ and $\tup b_1,\ldots,\tup b_q\in V^{\tup y}$, 
then in $\phi(G)$ we can find a semi-ladder $\tup c_1,\ldots,\tup c_\ell$ and $\tup d_1,\ldots,\tup d_\ell$, 
where $\tup c_i$-s are permutations of different tuples from $\{\tup a_1,\ldots,\tup a_q\}$, and $\tup d_i$-s are different tuples from $\{\tup b_1,\ldots,\tup b_q\}$.
For $j\in \{1,\ldots,k\}$ and a tuple $\tup a\in V^{\tup x}$, by $\tup a^j$ we denote the tuple from $V^{\tup x^j}$ that is obtained 
by permuting $\tup a$ as in the permutation that maps $\tup x$ to $\tup x^j$.

We proceed by induction on $\ell$.
For the base case $\ell=1$, we can take $\tup c_1=\tup a_1^1$ and $\tup d_1=\tup b_1$.

We move to the inductive step.
Since $\psi(\tup a_q;\tup b_t)$ holds for all $t<q$, there exists an index $j\in \{1,\ldots,k\}$ such that $\phi(\tup a_q^j;\tup b_t)$ holds 
for $\lceil \frac{q-1}{k}\rceil=k^{\ell-2}$ indices $t\in \{1,\ldots,q-1\}$.
Restricting sequences $\tup a_1,\ldots,\tup a_q$ and $\tup b_1,\ldots,\tup b_q$ to those indices $t$, we obtain a semi-ladder in $\psi(G)$ of order $k^{\ell-2}$.
By applying the inductive assumption to this semi-ladder, we can find
a semi-ladder $\tup c_1,\ldots,\tup c_{\ell-1}$ and $\tup d_1,\ldots,\tup d_{\ell-1}$ in $\phi(G)$, where $\tup c_i$-s are permutations of different tuples from $\{\tup a_1,\ldots,\tup a_q\}$
and $\tup d_i$-s are different tuples from $\{\tup b_1,\ldots,\tup b_q\}$, all of which satisfy $\phi(\tup a_q^j;\tup d_i)$.
It now remains to extend this semi-ladder by appending $\tup c_\ell = \tup a_q^j$ and $\tup d_\ell=\tup b_q$.
\end{proof}

\begin{lemma}\label{lem:prof-comp-powers}
The conclusion of \cref{lem:num-profiles} holds also when $\Cc=\Dd^s$ for any nowhere dense class $\Dd$ and fixed $s\in \N$, and when $\Cc$ is the class of map graphs.
\end{lemma}
\begin{proof}
Consider first the case when $\Cc=\Dd^s$, where $\Dd$ is nowhere dense.
Take any graph from $\Cc$, say $G^s$ for $G=(V,E)\in \Dd$.
Observe that for any two vertices $u,v\in V$ we have
$$\dist_{G^s}(u,v)=\left\lceil \frac{\dist_G(u,v)}{s}\right\rceil.$$
It follows that for every $r\in \N$, vertex $u\in V$, and vertex subset $S\subseteq V$, $\profile_r^{G^s,S}(u)$ is uniquely determined by $\profile_{rs}^{G,S}(u)$. Consequently,
$$\numprofiles{r}{\Cc}{m}\leq \numprofiles{rs}{\Dd}{m}\qquad\textrm{for all }m\in \N.$$
The claim now follows from \cref{lem:num-profiles} applied to the class $\Dd$ and radius parameter $rs$.

Consider now the case when $\Cc$ is the class of map graphs.
Take any map graph $G=(V,E)$.
By \cref{thm:half-squares}, there exists a bipartite planar graph $H$ such that one part of its bipartition is $V$ and $G$ is the subgraph of $H^2$ induced by $V$.
It follows that for any $u,v\in V$, we have $\dist_G(u,v)=\dist_H(u,v)/2$.
Therefore, for any $r\in \N$, $u\in V$, and $S\subseteq V$, $\profile_r^{G,S}(u)$ is uniquely defined by $\profile_{2r}^{H,S}(u)$, implying
$$\numprofiles{r}{\Cc}{m}\leq \numprofiles{2r}{\Pc}{m}\qquad\textrm{for all }m\in \N,$$
where $\Pc$ is the class of planar graphs. We may conclude as before using \cref{lem:num-profiles}, because $\Pc$ is nowhere dense.
\end{proof}

\begin{lemma}\label{lem:prof-comp-ktt}
Fix $t\in \N$ and let $\Cc$ be the class of $K_{t,t}$-free graphs. Then $\numprofiles{1}{\Cc}{m}\leq \Oh(m^t)$.
\end{lemma}
\begin{proof}
Let $G=(V,E)$ be a $K_{t,t}$-free graph and let $S\subseteq V$ be any subset of vertices with $|S|\leq m$.
Observe that for two vertices $u,v\in V\setminus S$, we have $\profile^{G,S}_1(u)=\profile^{G,S}_1(v)$ if and only if $N(u)\cap S=N(v)\cap S$.
Since there are at most $m$ vertices in $S$ itself, it suffices to bound the size of the set system $\Ff=\{N(u)\cap S\colon u\in V\setminus S\}$ by $\Oh(m^t)$.

To this end, we first note that $\Ff$ obviously contains at most $\binom{m}{0}+\binom{m}{1}+\ldots+\binom{m}{t-1}=\Oh(m^{t-1})$ sets of size smaller than $t$.
To bound the number of sets in $\Ff$ of size at least $t$, do the following.
For every vertex $u\in V\setminus S$ with at least $t$ neighbors in $S$, pick an arbitrary set $A_u$ consisting of $t$ neighbors of $u$ in $S$.
If any subset $A\subseteq S$ with $|A|=t$ was picked as $A_u$ for at least $t$ vertices $u\in V\setminus S$, then $A$ together with those vertices would form a $K_{t,t}$ subgraph in $G$, a contradiction.
Therefore, by the pigeonhole principle, the total number of vertices in $V\setminus S$ that have at least $t$ neighbors in $S$ is bounded by $(t-1)m^t$, yielding the same upper bound on the number of
sets in $\Ff$ of size at least $t$. Consequently, we have that $|\Ff|\leq \Oh(m^{t-1})+(t-1)m^t=\Oh(m^t)$.
\end{proof}

\bibliographystyle{abbrv}
\bibliography{ref}

\begin{thebibliography}{10}

\bibitem{adler2014interpreting}
H.~Adler and I.~Adler.
\newblock Interpreting nowhere dense graph classes as a classical notion of
  model theory.
\newblock {\em European Journal of Combinatorics}, 36:322--330, 2014.

\bibitem{Chen01a}
Z.~Chen.
\newblock Approximation algorithms for independent sets in map graphs.
\newblock {\em J. Algorithms}, 41(1):20--40, 2001.

\bibitem{ChenGP02}
Z.~Chen, M.~Grigni, and C.~H. Papadimitriou.
\newblock Map graphs.
\newblock {\em J. {ACM}}, 49(2):127--138, 2002.

\bibitem{courcelle1990monadic}
B.~Courcelle.
\newblock The monadic second-order logic of graphs. {I}. {R}ecognizable sets of
  finite graphs.
\newblock {\em Information and computation}, 85(1):12--75, 1990.

\bibitem{DawarK09}
A.~Dawar and S.~Kreutzer.
\newblock Domination problems in nowhere-dense classes.
\newblock In {\em {FSTTCS 2009}}, volume~4 of {\em LIPIcs}, pages 157--168.
  Schloss Dagstuhl---Leibniz-Zentrum f\"ur Informatik, 2009.

\bibitem{dawar2009parameterized}
A.~Dawar and S.~Kreutzer.
\newblock Parameterized complexity of first-order logic.
\newblock In {\em Electronic Colloquium on Computational Complexity, TR09-131},
  page~39, 2009.

\bibitem{DemaineFHT05}
E.~D. Demaine, F.~V. Fomin, M.~Hajiaghayi, and D.~M. Thilikos.
\newblock Fixed-parameter algorithms for $(k,r)$-center in planar graphs and
  map graphs.
\newblock {\em {ACM} Trans. Algorithms}, 1(1):33--47, 2005.

\bibitem{drange2016kernelization}
P.~G. Drange, M.~S. Dregi, F.~V. Fomin, S.~Kreutzer, D.~Lokshtanov,
  M.~Pilipczuk, M.~Pilipczuk, F.~Reidl, F.~{S{\'{a}}nchez Villaamil},
  S.~Saurabh, S.~Siebertz, and S.~Sikdar.
\newblock Kernelization and sparseness: the case of {D}ominating {S}et.
\newblock In {\em {STACS 2016}}, volume~47 of {\em LIPIcs}, pages 31:1--31:14.
  Schloss Dagstuhl---Leibniz-Zentrum f\"ur Informatik, 2016.
\newblock Full version available as arxiv preprint 1411.4575.

\bibitem{dvovrak2013testing}
Z.~Dvo{\v{r}}{\'a}k, D.~Kr{\'a}l, and R.~Thomas.
\newblock Testing first-order properties for subclasses of sparse graphs.
\newblock {\em J. ACM}, 60(5):36, 2013.

\bibitem{eickmeyer2016neighborhood}
K.~Eickmeyer, A.~C. Giannopoulou, S.~Kreutzer, O.~Kwon, M.~Pilipczuk,
  R.~Rabinovich, and S.~Siebertz.
\newblock Neighborhood complexity and kernelization for nowhere dense classes
  of graphs.
\newblock In {\em {ICALP 2017}}, volume~80 of {\em LIPIcs}, pages 63:1--63:14.
  Schloss Dagstuhl --- Leibniz-Zentrum f\"ur Informatik, 2017.
\newblock Full version available as arxiv preprint 1612.08197.

\bibitem{GroheKS17}
M.~Grohe, S.~Kreutzer, and S.~Siebertz.
\newblock Deciding first-order properties of nowhere dense graphs.
\newblock {\em J. {ACM}}, 64(3):17:1--17:32, 2017.

\bibitem{nevsetvril2010first}
J.~Ne{\v{s}}et{\v{r}}il and P.~Ossona~de Mendez.
\newblock First order properties on nowhere dense structures.
\newblock {\em The Journal of Symbolic Logic}, 75(03):868--887, 2010.

\bibitem{nevsetvril2011nowhere}
J.~Ne{\v{s}}et{\v{r}}il and P.~Ossona~de Mendez.
\newblock On nowhere dense graphs.
\newblock {\em European Journal of Combinatorics}, 32(4):600--617, 2011.

\bibitem{sparsity}
J.~Ne\v{s}et\v{r}il and P.~{Ossona de Mendez}.
\newblock {\em Sparsity --- {G}raphs, {S}tructures, and {A}lgorithms},
  volume~28 of {\em Algorithms and combinatorics}.
\newblock Springer, 2012.

\bibitem{PS18}
M.~Pilipczuk and S.~Siebertz.
\newblock Kernelization and approximation of distance-$r$ independent sets on
  nowhere dense graphs.
\newblock {\em arXiv preprint 1809.05675}, 2018.

\bibitem{notes}
M.~Pilipczuk and S.~Siebertz.
\newblock Lecture notes for the course ``{S}parsity'' given at {F}aculty of
  {M}athematics, {I}nformatics, and {M}echanics of the {U}niversity of
  {W}arsaw, Winter Semester 2017/18.
\newblock Available at \url{https://www.mimuw.edu.pl/~mp248287/sparsity}.

\bibitem{PilipczukST17}
M.~Pilipczuk, S.~Siebertz, and S.~Toru{\'{n}}czyk.
\newblock On the number of types in sparse graphs.
\newblock In {\em {LICS 2018}}, pages 799--808. {ACM}, 2018.

\bibitem{podewski1978stable}
K.-P. Podewski and M.~Ziegler.
\newblock Stable graphs.
\newblock {\em Fundamenta Mathematicae}, 100(2):101--107, 1978.

\bibitem{shelah1990classification}
S.~Shelah.
\newblock {\em Classification theory: and the number of non-isomorphic models},
  volume~92.
\newblock Elsevier, 1990.

\bibitem{TelleV12}
J.~A. Telle and Y.~Villanger.
\newblock {FPT} algorithms for domination in biclique-free graphs.
\newblock In {\em {ESA 2012}}, volume 7501 of {\em LNCS}, pages 802--812.
  Springer, 2012.

\bibitem{Thorup98}
M.~Thorup.
\newblock Map graphs in polynomial time.
\newblock In {\em {FOCS 1998}}, pages 396--405. {IEEE} Computer Society, 1998.

\end{thebibliography}

\end{document}